\documentclass{article}

\usepackage{fullpage}
\usepackage{amsmath}
\usepackage{amssymb}
\usepackage{natbib}

\usepackage{authblk}

\usepackage{hyperref}
\usepackage[capitalize]{cleveref}

\usepackage{thm-restate}

\def\proofname{proof}
\newtheorem{theorem}{Theorem}
\newtheorem{corollary}[theorem]{Corollary}

\newtheorem{lemma}[theorem]{Lemma}

\newtheorem{definition}[theorem]{Definition}
\newtheorem{proposition}[theorem]{Proposition}
\newtheorem{observation}[theorem]{Observation}
\newtheorem{remark}[theorem]{Remark}
\newtheorem{example}[theorem]{Example}

\newenvironment{proof}{\noindent\bf{Proof.}\rm}{\hfill$\blacksquare$\bigskip}

\usepackage[]{color-edits}
\addauthor{MB}{red}
\newcommand{\mbc}[1]{{\MBcomment{#1}}}

\addauthor{TE}{blue}

\addauthor{UF}{blue}
\newcommand{\ufc}[1]{{\UFcomment{#1}}}

\newcommand{\OLD}[1]{}

\newcommand{\items}{{\cal{M}}}
\newcommand{\agents}{{\cal{N}}}
\newcommand{\MMSi}{MMS_i} 
\newcommand{\TPSi}{TPS_i} 
\newcommand{\TPSiThreePar}[3]{\TPSi(#1,#2,#3)}
\newcommand{\TPSiThree}{\TPSiThreePar{n}{\items}{v_i}}
\newcommand{\PSi}{PS_i}

\newcommand{\inst}{{\cal{I}}}

\begin{document}
	
	\title{
	On Best-of-Both-Worlds Fair-Share Allocations}
	\author{Moshe Babaioff\thanks{Microsoft Research ---  E-mail: \texttt{moshe@microsoft.com}.}, Tomer Ezra\thanks{Sapienza University of Rome ---  E-mail: \texttt{tomer.ezra@gmail.com}.}, Uriel Feige\thanks{Weizmann Institute and Microsoft Research ---  E-mail: \texttt{uriel.feige@weizmann.ac.il}. }}
	\date{%
    \today
}

	\maketitle
	
	\begin{abstract}
		We consider the problem of fair allocation of indivisible items among $n$ agents with additive valuations, when agents have equal entitlements to the goods, and there are no transfers. Best-of-Both-Worlds (BoBW) fairness mechanisms aim to give all agents both an ex-ante guarantee (such as getting the proportional share in expectation) and an ex-post guarantee. 
Prior BoBW results have focused on ex-post guarantees that are based on the ``up to one item" paradigm, such as envy-free up to one item (EF1). In this work we attempt to give every agent a high \emph{value} ex-post, and specifically, a constant fraction of her maximin share (MMS). The up-to-one-item paradigm fails to give such a guarantee, and it is not difficult to present examples in which previous BoBW mechanisms give some agent only a $\frac{1}{n}$ fraction of her MMS.

Our main result is a deterministic polynomial-time algorithm that computes a distribution over allocations that is ex-ante proportional, and ex-post, every allocation gives every agent at least her proportional share up to one item, and more importantly, at least half of her MMS. Moreover, this last ex-post guarantee holds even with respect to a more demanding notion of a share, introduced in this paper, that we refer to as the \emph{truncated proportional share} (TPS). Our guarantees are nearly best possible, in the sense that one cannot guarantee agents more than their proportional share ex-ante, and one cannot guarantee {all agents value larger} than a $\frac{n}{2n-1}$-fraction of their TPS ex-post.
	\end{abstract}

	\section{Introduction}
	

In this paper we consider fair allocation of indivisible items to agents with additive valuations.
An \emph{instance} $\inst=(v,\items,\agents)$ of the fair allocation problem consists of a set $\items$ of $m$ indivisible items, 
a set $\agents$ of $n$ agents, 
and vector $v=(v_1,v_2,\ldots,v_n)$ of non-negative additive valuations, with the valuation of  agent $i\in \agents$ for set $S\subseteq \items$ being $v_i(S) = \sum_{j \in S} v_i(j)$, where $v_i(j)$ denotes the value of agent $i$ for item $j\in \items$.
We assume that the valuation functions of the agents are known to the social planer, and that there are no transfers (no money involved). We further assume that all agents have equal entitlement to the items.
An allocation $A$ is a collection of $n$ disjoint bundles  $A_1, \ldots, A_n$ (some of which might be empty), where $A_i \subseteq \items$ for every $i\in \agents$. A randomized allocation is a distribution over deterministic allocations. 
We wish to design randomized allocations that enjoy certain fairness properties.


Before discussing some standard fairness properties, we briefly motivate the {\em best of both worlds} (BoBW) framework, that considers both ex-ante and ex-post properties of randomized allocations. Consider a simple allocation instance $\inst_1$ with two agents and two equally valued items. Intuitively, {any fair allocation in this case is an allocation } 
that gives each agent one of the items. Giving both items to one of the agents and no item to the other agent is not considered fair. Consider now an instance $\inst_2$ with two agents and just one item. As we want to allocate the item (to achieve Pareto efficiency) but the item is indivisible, we give it to one of the agents, and then the other agent gets no item. The fact that some agent receives no item is unavoidable, and in this respect the allocation is fair. {Yet, the agent not getting the item might argue that this deterministic allocation is unfair as she has the same right to the item as the other agent. Indeed, we can improve the situation at least ex-ante: }
We can invoke a lottery to decide at random which of the two agents gets the item. {While for any realization inevitably one agent is left with nothing, the allocation mechanism is ex-ante fair} 
(each agent has a fair chance to win the lottery). Going back to instance $\inst_1$, we could also have a lottery for $\inst_1$, and have the winner receive both items. This too would be ex-ante fair, but ex-post (with respect to the final allocation) it would not be fair (as we did have the option to choose an allocation that gives every agent one item). Examples such as those above illustrate why {we want our allocation mechanism to concurrently enjoy \emph{both} ex-ante and ex-post fairness guarantees, as each guarantee by itself seems not to be sufficiently fair}.    

For the purpose of defining ex-ante fairness properties of randomized allocations, we assume that agents are risk neutral. That is, the ex-ante value that an agent derives from a distribution over bundles is the same as the expected value of a bundle selected at random from this distribution. Consequently, when considering a distribution $D$ over allocations (of $\items$ to $\agents$), we also consider the expectation of this distribution, which can be interpreted as a {\em fractional allocation}. In this fraction allocation, the fraction of item $i$ given to agent $j$ exactly equals the probability with which agent $i$ receives item $j$ under $D$. We naturally extend the additive valuation functions of agents to fractional allocations, by considering the expected valuation, that is, an additive valuation where the value of a fraction $q_j$ of item $j$ to agent $i$ is $q_j \cdot v_i(j)$. 


	\subsection{Brief Review of Terminology and Notation}
	We briefly review some properties of allocations from the literature, properties that are most relevant to the current work and to prior related work. 

We start with standard share definitions. 
The {\em proportional share} of agent $i$ is $PS_i = \frac{v_i(\items)}{n}$. We say that an allocation  $A = (A_1, \ldots, A_n)$ is {\em proportional} if every agent $i$ gets value at least $PS_i$ (that is, $v_i(A_i)\geq \frac{v_i(\items)}{n}= \PSi$), and a fractional (randomized) allocation is {\em ex-ante proportional} if she gets her proportional share in expectation. We say that an allocation $A$ is \emph{proportional up to one item (Prop1)} 
if for every agent $i$ it holds that $v_i(A_i)\geq PS_i - \max_{j \in \items\setminus A_i} v_i(j)$. 
The {\em maximin share} $MMS_i$ of agent $i$ 
is the maximum value that $i$ could secure if she was to partition $\items$ into $n$ bundles, and receive the bundle with the lowest value under $v_i$. 

We next discuss envy. An allocation is {\em envy free} (EF) if every agent (weakly) prefers her own bundle over that of any other agent, and a fractional (randomized) allocation is {\em ex-ante envy free} if for every agent, the expected value of her own allocation is at least as high as the expected value of the allocation of any other agent. 
Note that an allocation that is ex-ante envy free is ex-ante proportional.
An allocation is \emph{envy-free up to one good (EF1)} 
(\emph{envy-free up to any good (EFX)}, 
respectively) if every agent weakly prefers her own bundle over that of any other agent, up to the most (least, respectively) valuable item in the other agent's bundle.
Note that EF1 implies Prop1.
Finally, an allocation is \emph{envy-free up to one good more-and-less ($EF^1_1$)} 
if no agent $i$ envies another agent $j$ after removing one item from the set $j$ gets, and adding one item (not necessarily the same item) to $i$. 
Note that $EF^1_1$ is weaker than $EF1$.

Finally, we consider notions of efficiency. 
An (fractional) allocation \emph{Pareto dominates} another (fractional) allocation if it is weakly preferred by all agents, and strictly so by at least one. 
An integral allocation is \emph{Pareto optimal (PO)} if no integral allocation Pareto dominates it. 
An allocation (integral or fractional) is \emph{fractionally Pareto optimal (fPO)} if it is Pareto optimal, and moreover, no fractional allocation Pareto dominates it. 
Another notion of efficiency is that of \emph{Nash Social Welfare maximization}. 
The \emph{Nash Social Welfare (NSW)} of allocation $A = (A_1, \ldots, A_n)$ is $\left( \prod_{i \in \agents} v_i(A_i) \right)^{\frac{1}{n}}$. In case of fractional allocations, we use the notation fNSW.

	\subsection{Previous BoBW Results for Additive Valuations}
		\label{sec:previousBoBW}

The  state of the art BoBW results for additive valuations are presented in the two recent papers of  \citet{FSV20, Aziz20}.
Both of these works are based on the well known paradigm that we call here ``faithful implementation of a fractional allocation":  
	a distribution over deterministic allocations is a \emph{faithful implementation of the fractional allocation} if 
	the ex-ante (expected) value of every agent under the distribution is the same as it is in the fractional allocation, and ex-post (for any realization) it is the same as the expectation, up to the value of one item. 
	Both papers use versions of the result of \citet{BCKM13} showing that any fractional allocation can be  faithfully implemented. 
Various versions of these results were presented in the past, and in Appendix~\ref{sec:faithful} we survey those results. 
In Section~\ref{sec:preliminaries} we formally present a version of ``faithful implementation" that summarizes 
the prior results, stated as Lemma~\ref{lem:faithful}.

By ``faithful implementing" the fractional allocation that is the outcome of multiple executions of the probabilistic serial mechanism (a.k.a. {\em eating} mechanism)  of~\citet{BM01} till there are no more items, the following BoBW result was proved in~\citep{Aziz20}. (The same theorem was established earlier in~\citep{FSV20}, but with a somewhat more complicated proof.)


\begin{theorem}[\citep{FSV20, Aziz20}]
	\label{thm:Aziz20}
	There is a deterministic polynomial-time faithful implementation of a fractional allocation that is ex-ante envy free (and thus ex-ante proportional), and the implementation is supported on allocations that are (ex-post) EF1. 
\end{theorem}

By ``faithful implementing" the fractional allocation that maximizes the fractional Nash Social Welfare, the following BoBW result was proved in~\citep{FSV20}. 

\begin{theorem}[\citep{FSV20}]
	\label{thm:FSV20}
	There is a deterministic polynomial-time faithful implementation of a fractional allocation that is ex-ante fPO and ex-ante proportional, and the implementation is supported on allocations that are (ex-post) fPO, Prop1, and $EF^1_1$.
\end{theorem}

The ``up to one item" paradigm used in the ex-post guarantees of Theorems~\ref{thm:Aziz20} and~\ref{thm:FSV20} is most useful when a difference of one item does not make a big difference in value.
However, when items do have large values,  it does not guarantee  agents a high ex-post value. 
In contrast, we aim to give each agent ``high enough value" ex-post, where value is measure compared to ``what the agent deserves", captured by her fair share.
Specifically, we 	
aim to give every agent a large fraction (``an approximation") of her ``fair share", e.g. half the agent's MMS share.
The following allocation instance shows that neither Theorem~\ref{thm:Aziz20} nor Theorem~\ref{thm:FSV20} provide a constant approximation for the MMS ex-post, {and both are supported only on allocations that are intuitively very unfair}. 
Moreover, in this instance the MMS equals the proportional share, and hence one cannot dismiss this example as one in which the MMS is too small for the agents to care about.

Consider an instance with $n$ identical items, each of value $n$. In this case it is clear each agent should get one item ex-post. Now, {suppose that one} 
of those big items is split into $n$ small items, each of value 1. In this case we want one agent to get all of these small items, and each other agent to get one of the big items. Our next example shows that once these small items are not completely identical, but rather each agent slightly prefers a different one of them, 
then in both prior BoBW results, {in every realization},
one of the agents ends up 
getting only a small fraction of her MMS. 

\begin{example}
	\label{ex:notMMS}
	The instance has $2n-1$ items $\{s_1,s_2, \ldots , s_{n}\} \cup \{b_1,b_2, \ldots , b_{n-1}\}  $. 
	For some small $\epsilon > 0$, for every agent $i\in \agents$, the additive valuation function $v_i$ is as follows:
	
	\begin{itemize}
		\item $v_i(s_i) = 1 + \epsilon$. 
		\item {$v_i(s_j) = 1-\frac{\epsilon}{n-1}$} for every $1 \le j \le n$ such that $j \not= i$. 
		\item $v_i(b_j) = n$ for every $1 \le j \le n-1$. 
	\end{itemize}
	
	The MMS of every agent is $n$: one bundle contains all small items $\{s_1, \ldots, s_n\}$, and the remaining bundles each contain one of the remaining, big, items. The proportional share of every agent {is also $n$.} 
	
	{In every allocation, at least on agent does not receive a big item, as there are fewer big items than agents.
	The algorithms of Theorem~\ref{thm:Aziz20} gives every agent at most $\lceil \frac{|\items|}{n} \rceil$ items. Hence the agent that does not receive a big item receives a value of at most $2 - \epsilon$, whereas her MMS is $n$. 
	
	The algorithm of Theorem~\ref{thm:FSV20} starts with a fractional allocation 
	that maximizes the fractional Nash Social Welfare. This fractional allocation necessarily allocates the small item $s_i$ integrally to agent $i$, for every $i \le n$. Consequently, also ex-post, every agent $i$ gets the respective item $s_i$. 
	By the pigeon-hole principle, in an ex-post allocation there is an agent that receives no item among the big items $\{b_1,b_2, \ldots , b_{n-1}\} $. This agent $i$ receives only the small item $s_i$, and hence only a $\frac{1 + \epsilon}{n}$ fraction of her MMS.} 
\end{example}

	\subsection{Our Contributions}

In this paper we aim for a Best-of-Both-Worlds fairness result: a randomized allocation that gives every agent at least her proportional share ex-ante, and some guaranteed value ex-post. {The ex-post guarantee we give is at least half the MMS, and in fact, stronger.} We introduce  a new notion of share that we refer to as the {\em truncated proportional share (TPS)}, which we believe might be of independent interest. 
	We show that the TPS is at least as large as the MMS, and our BoBW result guarantees half of the TPS ex-post (and thus half the MMS ex post), while also giving each agent her proportional share ex-ante. 



\subsubsection{The Truncated Proportional Share}
We next define the Truncated Proportional Share of an agent with an additive valuation. 
	As we will see later, this share has two advantages over MMS: it is at least as high as the MMS, and while the MMS is NP-hard to compute, the TPS is easy to compute. 
	We alert the reader that in this paper we define TPS only with respect to additive valuation functions (while the definition of MMS extends without change beyond additive valuations).
\begin{definition}
	\label{def:TPS}
	For a setting with $n$ agents and a set of items $\items$, the \emph{truncated proportional share} $\TPSi= \TPSiThree$  of agent $i$ with additive valuation function $v_i$ is the largest value $t$ such that $\frac{1}{n} \sum_{j\in \items} \min[v_i(j), t] = t$. 
\end{definition}

{We note that the TPS is well defined, as $t=0$ satisfies the equality, and the maximum is obtained as the RHS is linear, while the LHS is piece-wise-linear with finitely many segments (at most $m$). 
}
From the definition of TPS it is immediate to see that $\TPSi \le \PSi$, but $\TPSi$ is smaller than $\PSi$ when there is at least one 
\emph{over-proportional item}, which is an item that by itself gives agent $i$ value larger than $\PSi$. 
	Observe that if the value of every over-proportional  items is reduced to $TPS_i$, then the TPS is the proportional share of the resulting valuation function after these reductions. 
	In absence of over-proportional items, clearly $TPS_i = PS_i$. Yet, when there are over-proportional items, $\TPSi$ might be much smaller than $\PSi$. For example, whenever there are less items than agents (e.g. a single item and two agents that desire it) then for every agent $\TPSi=\MMSi=0$ while $\PSi>0$. In any such case, it is clearly impossible to concurrently give all agents a positive fraction of their proportional share ex post (while the truncated proportional shares and the maximin share are small enough to make  their approximation plausible). 

Moreover, regardless of the presence of over-proportional items, $TPS_i \ge MMS_i$. This is because taking $t = MMS_i$ satisfies $\frac{1}{n} \sum_{j\in \items} \min[v_j, t] \ge t$ (as every one of the $n$ bundles in the partition that determines $MMS_i$ contributes at least $t$ to the sum), which implies that $t$ in Definition~\ref{def:TPS} is at least as large as $MMS_i$. Hence $MMS_i \le TPS_i \le PS_i$. In particular, guarantees with respect to the TPS imply at least the same guarantees with respect to MMS, and sometimes better. 

{The following example illustrates the TPS definition:}
{
\begin{example} \label{ex:tps}
There are $n=4$ agents and $m =5$ items.
The values of the items for agent $i$ are  $2,3,4,5,6$.
Her proportional share $\PSi$ equals $\frac{2+3+4+5+6}{4} = 5$, and her truncated proportional share $\TPSi$ can be seen to be $4.5$. Her TPS is at least $4.5$ since $\frac{2+3+4+4.5+4.5}{4} =4.5$.
Her TPS is at most $4.5$ since for every $t>4.5$, it holds that $\frac{2+3+4+\min(5,t) +\min(6,t)}{4} \leq \frac{9+2t}{4}<t$.
\end{example}
}


The TPS is a more tractable object than the MMS.
It is not difficult to see that 
the TPS can be recursively defined as follows: when $n=1$ then $\TPSi= \TPSiThreePar{1}{\items}{v_i} = v_i(\items)$, and when $n \ge 2$ then 
	$TPS_i$ is the minimum among $\frac{v_i(\items)}{n}$, the proportional share of agent $i$, and her TPS in a {\em reduced instance} in which an item $j$ of highest value is removed as well as one of the agents, that is, $\TPSi$ in this case is $\TPSiThreePar{n-1}{\items\setminus\{j\}}{v_i}$.
This procedure provides a simple polynomial time algorithm for computing the TPS: if the proportional share of the reduced instance is smaller than that of the original instance, compute $TPS_i$ for the reduced instance. If not, then $TPS_i$ is the proportional share of the original instance. (In contrast, computing the MMS is NP-hard.) {To demonstrate the calculation of the TPS, consider  Example~\ref{ex:tps}.
We first check whether her maximal value ($6$) is greater than her proportional share ($5$). Since $6>5$, we remove one agent and the maximal item, and calculate the TPS for the remaining items and agents.
The new proportional share is now $\frac{2+3+4+5}{3}=\frac{14}{3}$, while the maximal value is $5$ which is still greater. Thus, we  remove again one agent, and the maximal item.
Now, the proportional share is $\frac{2+3+4}{2}=4.5$ and the maximal value is $4$ which is at most her proportional, therefore $4.5$ is her TPS.}

Moreover, consider $\rho_{TPS}$, the highest fraction such that in every instance, there is an allocation giving every agent a $\rho_{TPS}$ fraction of her TPS. It is easy to determine the exact value of $\rho_{TPS}$, which turns out to be $\frac{n}{2n-1}$. 
(In contrast, the exact value of the corresponding $\rho_{MMS}$ is unknown~\citep{KurokawaPW18,GT20}.) To see that $\rho_{TPS} \ge \frac{n}{2n-1}$, we observe that a polynomial time allocation algorithm of~\citet{LMMS04} 
gives every agent a $\frac{n}{2n-1}$ fraction of her TPS. (For more details, see Appendix~\ref{app:proofs}.) To see that $\rho_{TPS} \le \frac{n}{2n-1}$, consider an instance with $2n-1$ items, each of value~1. The TPS of every agent is $\frac{2n-1}{n}$, but in every allocation, at least one of the agents gets at most one item, and hence value at most~1. 

The example above also shows that the TPS of an agent can be factor $\frac{2n-1}{n}$ larger than her MMS.
{This ratio is tight, because $MMS_i \ge \frac{n}{2n-1} TPS_i$ for every agent $i$. This follows by considering $n$ agents with the same valuation function $v_i$,
and recalling that there is an allocation that gives every agent at least a $\frac{n}{2n-1}$ fraction of her TPS. The $n$ bundles of this allocation each have a value of at least $\frac{n}{2n-1} TPS_i$, and hence they form a partition of $\items$ that shows that $MMS_i \ge \frac{n}{2n-1} TPS_i$.} 
{We summarize the above discussion in the following proposition:}
\begin{proposition}
	For any setting with $n$ agents and any additive valuation $v_i$ it holds that \\
	$$
	PS_i \geq TPS_i \geq 	MMS_i \geq  \tfrac{n}{2n-1} \cdot TPS_i
	$$
	Moreover, each of the above inequalities is strict for some instance, and holds as equality for some other instance.  
\end{proposition}



\subsubsection{Our Best-of-Both-Worlds Result}
We now return to present our main result. Due to the difficulties alluded to in Example~\ref{ex:notMMS} and Proposition~\ref{pro:noParetoExAnte}, we do not follow the paradigm of starting with a simple to describe fractional allocation, and then faithfully implementing it (using Lemma~\ref{lem:faithful}). 
Instead, we design an algorithm that generates a distribution over allocations that each gives every agent at least half of her TPS, with the additional property that every agent gets at least her proportional share in expectation. Along the way, we do use Lemma~\ref{lem:faithful}, but we apply it on fractional allocations that involve only carefully selected subsets of $\items$, rather than a fractional allocation that involves all of $\items$. 
Our main result is the following.

\begin{restatable}{theorem}{thmBoBW}		
	\label{thm:BoB}
	For every allocation instance with additive valuations, there is a randomized allocation that is ex-ante proportional, 
	and gives each agent at least half of her TPS ex-post (and hence also at least half of her MMS), as well as being Prop1 ex-post. Moreover, there is a deterministic polynomial time algorithm that, given the valuation functions of the agents, computes 
	such a randomized allocation, supported on at most $n$ allocations.
\end{restatable}

Theorem~\ref{thm:BoB} is nearly the best possible {in the following senses}.
{First,} it is not possible to guarantee every agent value that is strictly larger than her proportional share ex-ante (e.g., if all agents have the same valuation function). 
{Second,} the highest possible fraction of the truncated proportional share that can be guaranteed ex-post is {at most} $\frac{n}{2n-1} = \frac{1}{2} + \frac{1}{4n-2}$ ({recall the example above with the $2n-1$ identical items}), which tends to half as $n$ grows large, and the theorem indeed ensures a fraction of half. We also remark that for the instance in Example~\ref{ex:notMMS}, while in the BoBW results from prior work \citep{FSV20,Aziz20} there is always an agent that gets {only a small} fraction of her MMS,
the algorithm of Theorem~\ref{thm:BoB} gives every agent her TPS (and her MMS) ex post. 

Another aspect in which Theorem~\ref{thm:BoB} cannot be improved is with respect to its Pareto properties. While the prior result of \citet{FSV20}  present a BoBW result (Theorem \ref{thm:FSV20}) with a distribution over allocations that is ex-ante fPO, our result does not give ex-ante fPO. 
We next show that if we want every agent to receive ex-post at least a constant fraction of her maximin share, getting the guarantee of ex-ante fPO is impossible. Moreover, this conflict between ex-ante fPO and half the MMS concerns every ex-post allocation that might potentially be in the support, not just one of them.


\begin{restatable}{proposition}{propEasyBoBW}
	\label{pro:noParetoExAnte}
	For every $n \ge 2$ and every $\epsilon > 0$ there are allocation instances with additive valuations, with the following property: for every ex-ante Pareto optimal (fPO) randomized allocation (whether ex-ante proportional or not), every allocation in its support does not give some agent more than a $\frac{1 + \epsilon}{n}$ fraction of her maximin share. 
\end{restatable}

The proof of Proposition~\ref{pro:noParetoExAnte} is based on Example~\ref{ex:notMMS}, and is given in Appendix \ref{app:proofs}. 


We thus see that we cannot hope to improve our result to also guarantee  ex-ante fPO. How about the weaker condition of ex-post PO?
The polynomial time algorithm referred to in Theorem~\ref{thm:BoB} does not necessarily produce Pareto efficient allocations. However, the existential result in the theorem does hold simultaneously with a Pareto efficiency requirement, for the simple reason that ex-post replacement of an allocation by an allocation that Pareto dominates it cannot reduce the received fraction of the (ex-post) truncated proportional share (and ex-ante proportional share) of any of the agents. It is not clear whether this reallocation can be done in polynomial time.  (For NP-hardness results associated with Pareto efficient reallocation, see~\citep{dKBKZ09,ABLLM16}.)

\begin{corollary}
	\label{cor:BoB}
	For every allocation instance with additive valuations, there is a randomized allocation that is supported on at most $n$ allocations, is ex-ante proportional, and ex-post it gives every agent at least half of her TPS (and hence also at least half her MMS), as well as being ex-post PO.\footnote{The improvement to a PO allocation might not maintain the Prop1 property, yet each agent's value never decreases under that improvement.  } 
\end{corollary}

	\subsection{Additional Related Work}
	
{The maximin share was introduced by \citet{Budish11} as a relaxation of the proportional share. \citet{KurokawaPW18} showed that for agents with additive valuations, an allocation that gives each agent her MMS may not exist.
A series of papers
[\citealp{KurokawaPW18}, \citealp{amanatidis2017approximation}, \citealp{BK20}, \citealp{GhodsiHSSY18}, \citealp{garg2019approximating} ,\citealp{GT20}]
considered 
the best fraction of the MMS that can be concurrently guaranteed to all agents, and the current state of the art (for additive valuations) is a {$\frac{3}{4} +\Theta(\frac{1}{n}) $-fraction} 
of the MMS {(whereas there are instances in which more than a $\frac{39}{40}$ fraction of the MMS cannot be achieved~\citep{FST21})}.
For the case of arbitrary (non-equal) entitlements,
\citet{babaioff2021fair} define a share named the \emph{AnyPrice share (APS)}, which is 
the value the agent can guarantee herself whenever her budget is set to her entitlement  $b_i$ (when $\sum_i b_i=1$) and she buys her highest value affordable set when items are {adversarially} priced with a total price of $1$.  
To approximate the APS, they 
extend our definition of the TPS to the case of unequal entitlements.
For additive valuations they show that $TPS_i\geq APS_i 
$ and that the inequality is strict for some instances.

}

The fairness notion of 
Prop1 was introduced by \citet{conitzer2017fair}.
The fairness notion of EF1 was implicitly used in \citep{LMMS04}, and was formally defined in \citep{Budish11}. EFX (envy-free up to any good) was introduced in
\citep{CKMPSW19} 
The notion of envy-free up to one good more-and-less ($EF^1_1$) was defined in \citep{BarmanK19}, relaxing EF1.

For a subclass of additive valuations, that of additive {\em dichotomous} valuations, very strong BoBW results are known~\citep{HPPS20,Aziz20,BEF20}, which among other properties, are EF ex-ante, EFX ex-post, maximize welfare, and the underlying allocation mechanism is universally truthful. Such a strong combination of results is impossible to achieve for general additive valuations. 
In particular, the results of \citet{ABCM2017} imply that  every universally truthful randomized allocation mechanism for two agents that allocates all items must sometimes not give an agent more than a $\frac{2}{m}$ fraction of her MMS ex-post. 
{See Section~\ref{sec:truthful} for a more extensive discussion on truthfulness.}

In Section~\ref{sec:previousBoBW} we already discussed some previous BoBW results.
We further remark that in \citep{FSV20} they present an instance for which there is no randomized allocation that is ex-ante proportional, ex-post EF1 and ex-post fPO. For the same instance, there is no randomized allocation that is ex-ante proportional, is ex-post fPO, and gives every agent a positive fraction of her MMS.
	
	\section{Preliminaries}
	\label{sec:preliminaries}

We consider fair allocation of indivisible items to agents with additive valuations.
An \emph{instance} $\inst=(v,\items,\agents)$ of the fair allocation problem consists of a set $\items$ of $m$ indivisible items, a set $\agents$ of $n$ agents, and vector $v=(v_1,v_2,\ldots,v_n)$ of non-negative additive valuations, with the valuation of  agent $i\in \agents$ for set $S\subseteq \items$ being $v_i(S) = \sum_{j \in S} v_i(j)$, where $v_i(j)$ denotes the value of agent $i$ for item $j\in \items$.
We assume that the valuation functions of the agents are known to the social planer, and that there are no transfers (no money involved). We further assume that all agents have equal entitlement to the items.
An allocation $A$ is a collection of $n$ disjoint bundles  $A_1, \ldots, A_n$ (some of which might be empty), where $A_i \subseteq \items$ for every $i\in \agents$. 

As we shall be dealing with randomized allocations, let us introduce terminology that we shall use in this context. A \emph{random allocation} is a distribution $D$ over integral allocations $A^1, A^2, \ldots$. It induces an {\em expected} allocation $A^*$, where $A^*_{ij}$ specifies for agent $i$ and item $j$ the probability that agent $i$ receives item $j$, when an allocation is chosen at random from the underlying distribution $D$. These probabilities can be interpreted as fractions of the item that an agent receives ex-ante. Hence the expected allocation $A^*$ can be viewed as a \emph{fractional allocation}, in which items are divisible. Conversely, we say that the distribution $D$ (namely, the random allocation) {\em implements} the fractional allocation $A^*$ {when the expectation of $D$ is $A^*$}. 
Finally, we note that an additive valuation function can be extended in a natural way from allocations to fractional allocations, by considering the expected valuation. 
That is, the value of a $p_j$ fraction of item $j$ to agent $i$ to is $p_j\cdot v_i(j)$, and the value of a fractional allocation $A^*$ to agent $i$ is $\sum_{j\in \items} A^*_{ij}\cdot v_i(j)  $. 

For the issue of computing randomized allocations there are two different notions of polynomial time computation. 
In a \emph{random polynomial time implementation}, there is a randomized polynomial time algorithm that samples an allocation from the distribution $D$. In a \emph{polynomial time implementation}, there is a deterministic polynomial time algorithm that lists all allocations in the support of $D$ (implying in particular that the support contains at most polynomially many allocations), together with their associated probabilities. 


\subsection{Faithful Implementation}\label{sec:faithful-imp}
For general additive valuations, there is a very useful lemma that greatly simplifies the design of BoBW allocations. We refer to it here as the {\em faithful implementation lemma}. The lemma (sometimes with slight variations) was previously stated and used in BoBW results~\citep{BCKM13, FSV20, HPPS20, Aziz20}, and was used even earlier in approximation algorithms for maximizing welfare~\citep{Srinivasan08}. Restricted variants of it were introduced for scheduling problems~\citep{LST90}, and were later used for allocation problems~\citep{BD05}. For an extensive discussion of the faithful implementation lemma, as well as its proof (presented for completeness), see Appendix~\ref{sec:faithful}. 

\begin{restatable}{lemma}{lemmaFaithful}	
\label{lem:faithful}
Let $A^*$ be a fractional allocation of $m$ items to $n$ agents with additive valuations, and let $f$ denote the number of strictly fractional variables in $A^*$ (number of pairs $(i,j)$ such that in $A^*$, the fraction of item $j$ allocated to agent $i$ is strictly between 0 and 1). Then there is a deterministic polynomial time implementation of $A^*$, supported only on allocations in which every agent gets value (ex-post) equal her ex-ante value (in the fractional allocation $A^*$), up to the value of one item. (For agent $i$, the corresponding one item is the item most valuable to $i$, among those items that are assigned to $i$ under $A^*$ in a strictly fractional fashion. {Moreover, the values that the agent gets in any two allocations differ by at most the value of this single item.}) The distribution of the implementation is supported over at most $f + 1$ allocations. 
\end{restatable}

Using Lemma~\ref{lem:faithful}, one trivially gets the following BoBW result (implicit in previous work), which is a  baseline against which other BoBW results can be compared. 

\begin{proposition}
\label{pro:uniform}
There is a deterministic polynomial time implementation of a fractional allocation that is ex-ante envy free, and the implementation is supported on allocations that are (ex-post) Prop1.
\end{proposition}

\begin{proof}
Consider the {\em uniform fractional allocation}, that assigns a fraction of $\frac{1}{n}$ of every item to every agent. It is ex-ante envy free, as all agents get the same fractional allocation. Applying Lemma~\ref{lem:faithful}, it is implemented in deterministic polynomial time by allocations that are Prop1.
\end{proof}

	\section{Main Result: The Best of Both Worlds}
		\label{sec:proofs}

We now restate and prove our main result, Theorem \ref{thm:BoB}.

\thmBoBW*

In Section~\ref{sec:discussion} we discuss possible extensions of the theorem with regard to three aspects: fairness, efficiency and truthfulness. We present impossibilities of some natural extensions, as well as some open problems. 
 
\subsection{Proof Overview}
Let  $\inst=(v,\items,\agents)$ be an input instance. 
For any instance $\inst$ we denote  the proportional share and the truncated proportional share of agent $i$ by $\PSi(\inst)$ and $\TPSi(\inst)$ respectively.\footnote{Note that $\PSi$ and $\TPSi$ depend on $v_i$, but not on the valuations of the other agents.    }
For the original instance, we omit the instance and denote the proportional share and the truncated proportional share of agent $i$ by  $\PSi$ and $\TPSi$, respectively.

To prove the theorem we present a deterministic polynomial time algorithm that, given the input instance  $\inst=(v,\items,\agents)$, 
computes an implementation of a randomized allocation that gives every agent at least her proportional share ex-ante, and at least half of her truncated proportional share ex-post, 
and is supported on at most $n$ allocations.
Items that each by itself gives an agent her TPS will play a central rule in our algorithm. 
We say that item $j$ is {\em  exceptional} for agent $i$ if $v_i(j) \ge TPS_i$. 
Our algorithm has several phases:
\begin{enumerate}
	\item Find a distribution over $4n$ matchings. Each of these matchings partitions the agents to two disjoint sets $\agents_1$ and $\agents_2$, and the items to three disjoint sets $\items(\agents_1), \items(\agents_2)$ and $\items_3$ ($\items = \items(\agents_1)\cup  \items(\agents_2) \cup \items_3$, $|\agents_1|= |\items(\agents_1)|$ and $|\agents_2|= |\items(\agents_2)|$). Each agent in $\agents_1$ is matched with an item in $\items(\agents_1)$, and each agent in $\agents_2$ is matched with an item in $\items(\agents_2)$.
	The distribution over these matchings 
	is computed in two steps:
\begin{enumerate}
	\item \label{step:match-exceptional}	Compute (using LP1) a distribution in which in every matching, 
	every agent in $\agents_1$ is matched to an item that is  exceptional item for him in $\items(\agents_1)$ and such that:
	\begin{enumerate}
		\item No unallocated items (items in $\items\setminus \items(\agents_1)$) is  exceptional to any agent in $\agents_2=\agents\setminus \agents_1$.
		\item The distribution over these $4n$ matchings 
		gives each agent her proportional share conditioned on every agent in $\agents_2$ eventually getting her TPS in expectation (as indeed is guaranteed by \ref{step:1b2} below). 
	\end{enumerate}
	\item \label{step:match-others} Complete each partial matching to a complete matching by matching each agent in $\agents_2$ to an item in $\items(\agents_2)$,
	such that:
	\begin{enumerate}
		\item  Each agent prefers the item matched to him over any unmatched item (item in $\items_3$).
		\item \label{step:1b2}  For the unmatched items, there still is a fractional allocation of $\items_3$ such that for each agent in $\agents_2$, her expected value for the combination of her matched item and her fractional allocation is at least her TPS.
	\end{enumerate}    
\end{enumerate}

\item \label{step:LP3} For each matching above,
find (by LP3) a distribution over $m+1$ deterministic allocations that allocate $\items_3$, the unmatched items,  with the following properties:
\begin{enumerate}
	\item in each allocation in the support, every agent in $\agents_2$ gets (in total, over the matched item and the remaining allocation) at least half her TPS.
	\item In expectation, every agent in $\agents_2$ gets (in total) her TPS.
\end{enumerate}

\item \label{step:support} From the distribution over $4n(m+1)$ allocations defined above, find a distribution over at most $n$ of these allocations that is ex-ante proportional (all ex-post properties are preserved).  

\end{enumerate}

Before elaborating on these steps we make the following remark.

\begin{remark}
	Below we present an allocation algorithm with properties as in the theorem. 
	Some components in this algorithm are flexible. In particular, this applies to the objective functions of LP1 and LP3. Making use of this flexibility, we set the objective functions of these LPs so that they maximize welfare (subject to the constraints of the respective LPs). We find it natural to measure welfare in units of ``proportional share". Hence we assume (without loss of generality) that in $\inst$ the valuation function of every agent $i$ is scaled so that $v_i(\items) = n$. Consequently, $PS_i(\inst) = 1$ (and $TPS_i(\inst) \le 1$). This assumption is not used in the proof of Theorem~\ref{thm:BoB}, and is presented in this remark only so as to explain our particular choice of objective functions for LP1 and LP3. 
\end{remark}

\subsection{The Algorithm and the Proof}
We next move to formally describe all the steps of the algorithm and prove the theorem. 

{\bf Phase \ref{step:match-exceptional}:} Maximal allocation of  exceptional items. 

We start by transforming the input instance $\inst$ into a new instance $\inst_1$.
In $\inst_1$, we add $n$ auxiliary items to $\items$, and denote them by $a_1, \ldots, a_n$, thus obtaining a set $\items_1 =  \items \cup \{a_1, \ldots, a_n\}$. For every $i\in \agents$, we modify the original valuation function $v_i$ to the following {\em unit demand} valuation function $u_i$. 

\begin{itemize}
	
	\item For every item $j \in \items$, if $v_i(j) \ge \TPSi(\inst)$ then $u_i(j) = v_i(j)$.     
	
	\item For every item $j \in \items$, if $v_i(j) < \TPSi(\inst)$ then $u_i(j) = 0$.
	
	\item $u_i(a_i) = \TPSi(\inst)$. 
	
	\item $u_i(a_j) = 0$ for $j \not= i$.
	
	\item $u_i$ is {\em unit demand}. Namely, for $u_i(S) = \max_{j\in S} u_i(j)$ for every $S \subset \items_1$. 
	
\end{itemize}

We now set up a linear program that finds a fractional allocation that maximizes welfare in $\inst_1$, subject to the constraint that the fractional value received by every agent $i$ is at least $PS_i(\inst)$ (hence at least~1, due to our scaling). Variable $x_{ij}$ denotes the fraction of item $j$ received by agent $i$. Variable $s_i$ denotes the value that agent $i$ derives from the fractional allocation. We refer to the following linear program as LP1.

{\bf Maximize $\sum_{i\in \agents} s_i$ subject to:}

\begin{enumerate}
	
	\item $\sum_{i\in \agents} x_{ij} \le 1$ for every item $j\in \items$. (Every item is fractionally allocated at most once.)
	
	\item $\sum_{j\in \items_1} x_{ij} = 1$ for every agent $i\in \agents$. (Agent $i$ gets item fractions that sum to one item).
	
	\item $s_i = \sum_{j\in \items_1} u_i(j) x_{ij}$ for every agent $i\in \agents$. (Agent's $i$ value is the sum of fraction of values that she receives from the fractional allocation.)
	
	\item $s_i \ge PS_i(\inst)$ for every agent $i\in \agents$. (Agent's $i$ value is at least as high as $PS_i(\inst)$.)

	\item $x_{ij} \ge 0$ for every agent $i\in \agents$ and item $j\in \items$. 
	
\end{enumerate}

\begin{proposition}
	LP1 is feasible. 
\end{proposition}

\begin{proof}
	Let $E_i$ denote the set of items that are  exceptional for agent $i$ in the original instance $\inst$. {As there cannot be more than $n$ items worth more than the proportional share, we have that $|E_i| \le n$.} 
	Consider a solution for LP1 with $x_{ij} = \frac{1}{n}$ for every $i \in \agents$ and $j\in E_i$, and $x_{ia_i} = 1 - \sum_{j \in E_i} x_{ij}
	=1-\frac{|E_i|}{n}$. It clearly satisfies constraints~1,2,3 and~5. 
	In remains to establish that this solution  satisfies constraint~4, that is, the constraint $s_i \ge PS_i(\inst)$. 
	
	We shall use the facts that $PS_i(\inst) = \frac{1}{n}v_i(E_i) + \frac{1}{n}v_i(\items \setminus E_i)$ and $TPS_i(\inst) = \frac{v_i(\items \setminus E_i)}{n - |E_i|}$. 
	We can see that constraint~4 is satisfied by this solution: 
		$$s_i = \frac{1}{n}u_i(E_i) + \left(1 - \frac{|E_i|}{n}\right)u_i(a_i) =  \frac{1}{n}v_i(E_i) +  \frac{n - |E_i|}{n}TPS_i =  \frac{1}{n}(v_i(E_i) + v_i(\items \setminus E_i)) =PS_i(\inst) $$	
\end{proof}

We solve LP1. Let $A^*$ denote the optimal fractional solution that is found. 
We assume without loss of generality that ties are broken toward real items: In  $A^*$ there is no agent $i$ that is fractionally allocated her auxiliary item $a_i$ ($x_{ia_i}>0$) such that there is some real item $j\in \items$ with value $v_i(j)=\TPSi(\inst)$ that is not fully allocated (it is easy to make sure that is the case by shifting  mass from the auxiliary item to such a real item if not.) Moreover,
without loss of generality, in $A^*$ there are at most $2n-1$ items that are (fractionally) allocated.
(By constraint~2 the sum of allocated fractions is exactly $n$. Hence either there are $n$ items that are fully allocated, or there are fewer than $n$ fully allocated items. In the latter case, and only in the latter case, in addition there are items that are partly allocated. The number of partly allocated items need not exceed $n$, because no agent needs to receive fractions from two different partly allocated items. The agent can instead gradually increase her share in the more valuable of the two partly allocated items, and reduce her share in the other. This gradual increase stops when the first of the following two events happens: either constraint~1 becomes tight for the more valuable item, and then the item becomes fully allocated instead of partly allocated, or the share of the agent in the less valuable item becomes~0, and then the agent receives fractions from one less item.) Consequently, constraint~1 involves at most $2n-1$ items, and constraints~2,~3 and~4 contribute $3n$ additional constraints. As all $n$ variables $s_i$ are positive, the number of positive $x_{ij}$ variables is at most $4n-1$.

We perform faithful randomized rounding on $A^*$. The following proposition follows immediately from the properties of $A^*$ and Lemma~\ref{lem:faithful}, and hence its proof is omitted.

\begin{proposition}
	\label{pro:firstFaithful}
	The faithful randomized rounding of $A^*$ produces a distribution over allocations with the following properties:
	
	\begin{enumerate}
		
		\item In every allocation, every agent $i$ gets exactly one item from $\items_1$. This item is either one of her  exceptional items, or her auxiliary item $a_i$. {In either case the $u_i$ value of that item is least $TPS_i(\inst)$.} 
		
		\item The distribution is supported on at most $4n$ allocations. 
		
		\item In expectation, every agent $i$ gets value $s_i$ with respect to $u_i$. Recall that $s_i \ge PS_i(\inst)$. 
		
	\end{enumerate}
	
\end{proposition}

Consider now an arbitrary allocation $A'$ in the support of the faithful randomized rounding of $A^*$. With respect to $A'$, let $\agents_1 =\agents_1(A')$  denote the set of agents that receive an item that was  exceptional for them, and let $\agents_2$ denote the set of agents that receive their auxiliary item. (Note that $\agents_1 \cup \agents_2 = \agents$.) 

The first phase ends by giving each agent of $\agents_1$ the item that she receives under $A'$, and not giving agents of $\agents_2$ any item (as the auxiliary items do not really exist).
Thus we have that $\items(\agents_1)$ is the set of items matched to agents in $\agents_1$.
Observe that every agent $i\in \agents_1$ gets at least $TPS_i$ ex-post. The remaining phases will ensure that agents in $\agents_2$ get at least half their TPS ex-post. They will also ensure that ex-ante, every agent gets at least her proportional share (this will make use of item~3 of Proposition~\ref{pro:firstFaithful}).
\medskip

{\bf Phase \ref{step:match-others}:} Completing the matching.

If $\agents_2$ is empty, we go directly to  Phase~\ref{step:LP3}. Hence here we assume that $\agents_2$ is non-empty. 

Let $\items_2 \subset \items$ denote the subset of original items (not including the auxiliary items) that remain unallocated in $A'$ (those items not allocated to $\agents_1$).
Let $\inst_2$
denote the allocation instance that has $\items_2$ as its set of items, $\agents_2$ as its set of agents, and the valuation function of every agent $i \in \agents_2$ remains $v_i$ (restricted to the items in $\items_2$.) As $\inst_2$ is obtained from $\inst$ by removing $|\agents_1|$ agents and $|\agents_1|$ items, it holds that $TPS_i(\inst_2) \ge TPS_i(\inst)$ for every agent $i \in \agents_2$. 

Importantly, recall that  we may assume without loss of generality that $\items_2$ has no item that according to instance $\inst$ was  exceptional for an agent of $\agents_2$. (If  $\items_2$ contains an item $j$ that is exceptional for $i \in \agents_2$, then in $A'$, give $j$ instead of $a_i$ to agent $i$, by this moving agent $i$ out of $\agents_2$ and into $\agents_1$.) The fact that $TPS_i(\inst_2) \ge TPS_i(\inst)$ (for $i \in \agents_2$) implies that also in $\inst_2$,  $\items_2$ has no item that is  exceptional for an agent of $\agents_2$. Consequently, we infer that for every agent $i \in \agents_2$:

\begin{itemize}
	
	\item $TPS_i(\inst_2) = \frac{v_i(\items_2)}{|\agents_2|}$, and consequently also $\frac{v_i(\items_2)}{|\agents_2|} \ge TPS_i(\inst)$. 
	
	\item 
	{There are strictly more than $|\agents_2|$ items $j \in \items_2$ with $v_i(j) > 0$.} (This holds because $v_i(j) < TPS_i(\inst_2)$ for every $j \in \items_2$.)

\end{itemize}

Let $B_i \subset \items_2$ denote the set of $|\agents_2|$ items of highest value to agent $i \in \agents_2$, breaking ties arbitrarily. Let $W_i = v_i(B_i)$. As $\items_2$ has more than $|\agents_2|$ items of positive value for $i$, it follows that $W_i < v_i(\items_2)$.

We now transform the instance $\inst_2$ into a new instance $\inst'_2$. The set of items in $\inst'_2$ is $\items_2$, and the set of agents is $\agents_2$. Every agent $i\in \agents_2$ has a {\em unit demand} valuation function $w_i$, defined as follows.  For $j \in B_i$ we have $w_i(j) = \frac{v_i(j)}{v_i(\items_2) - W_i}$ (observe that the denominator is positive), and for $j \not\in B_i$ we have $w_i(j) = 0$. 

In the matching completion phase, we find a welfare maximizing allocation $B^*$ in $\inst'_2$. Observe that this can be done in polynomial time, because agents are unit demand, and hence finding $B^*$ amounts to solving an instance of maximum weight matching in a bipartite graph $G$, with $\agents_2$ as the set of left side vertices, $\items_2$ as the set of right side vertices, and weight $w_i(j)$ for edge $(i,j)$. In $B^*$, every agent $i \in \agents_2$ receives an item from her respective set $B_i$ (this follows because $|B_i| \ge |\agents_2|$).
We have that $\items(\agents_2)$ is the set of items matched to agents in $\agents_2$.

By the end of the matching phase, every agent holds one item. Agents in $\agents_1$ received their item in Phase~\ref{step:match-exceptional} (under $A'$), whereas agents in $\agents_2$ received their item in Phase~\ref{step:match-others} (under $B^*$). Let $e_i$ denote the item that has been allocated to agent $i$, and let $\items_3 = \items \setminus \{e_1, \ldots, e_n\}$ denote the set of items that are not yet allocated. A key property established by the first two phases is summarized in the following proposition.

\begin{proposition}
	\label{pro:firstItem}
	For every agent $i\in \agents$ it holds that $v_i(e_i) \ge \max_{j\in \items_3} v_i(j)$.
\end{proposition}

\begin{proof}
	For an agent $i\in \agents_1$, the proposition follows from the optimality of the fractional allocation $A^*$. If there is an item {$j \in \items_3$} with $v_i(j) > v_i(e_i)$, then in $A'$ item $j$ could replace item $e_i$ for agent $i$, thus increasing the welfare of $A'$. This would imply that LP1 has a fractional solution of value higher than that of $A^*$, contradicting the optimality of $A^*$.
	
	For an agent $i\in \agents_2$, the proposition follows from the optimality of the integral allocation $B^*$.
\end{proof}

We are now ready to move to the next phase of our algorithm.
\medskip

{\bf Phase~\ref{step:LP3}:} Allocating unmatched items.

In this phase, for each matching computed before, we allocate the items of $\items_3$, the items not in the matching. Every agent $i \in \agents$ has her original valuation function $v_i$ (with $v_i(\items) = n$). 

We first compute a fractional allocation for the items of $\items_3$. This is done by solving a linear program that we refer to as LP3. In LP3, variable $x_{ij}$ denotes the fraction of item $j\in \items_3$ allocated to agent $i$, and $s_i$ denotes the value that agent $i$ derives from the fraction allocation (under valuation function $v_i$).  The parameters $f_i$ are treated as constants in LP3. Their values are computed based on Phase~\ref{step:match-others}. Specifically,  $f_i = \frac{\frac{v_i(\items_2)}{|\agents_2|} - v_i(e_i)}{v_i(\items_2) - W_i}$ for $i\in \agents_2$ (where $e_i$ is the item allocated to agent $i$ in Phase~\ref{step:match-others}, and $W_i = v_i(B_i)$, as defined in Phase~\ref{step:match-others}).
We now present LP3.

{\bf Maximize $\sum_{i\in \agents} s_i$ subject to:}

\begin{enumerate}
	
	
	\item $\sum_{i\in \agents} x_{ij} \le 1$ for every item $j \in \items_3$. (Every item is fractionally  allocated at most once.)
	
	\item $s_i = \sum_{j\in \items_3} v_i(j) x_{ij}$ for every agent $i\in \agents$. (Agent's $i$ value is the sum of the fractions of values that she receives from the fractional allocation.)
	
	\item $s_i \ge f_i v_i(\items_3)$ for every agent $i\in \agents_2$. (This is the key constraint that ties LP3 with the allocation $B^*$ of Phase~\ref{step:match-others}. It applies only to agents in $\agents_2$.)
	
	\item $x_{ij} \ge 0$ for every agent $i\in \agents$ and item $j\in \items_3$. 
	
\end{enumerate}

Note that LP3 may fractionally allocate items from $\items_3$ to agents in $\agents_1$, but only after each agent in $\agents_2$ receives items of sufficiently high value as dictated by  Constraint~3.

\begin{lemma}
	\label{lem:Phase3}
	LP3 is feasible.
\end{lemma}

\begin{proof}
	Constraints~2 and~4 are satisfied by every solution in which $x_{i,j} \ge 0$ (for all $i$ and $j$). It remains to show that constraints~1 and~3 can be satisfied simultaneously.

	Recall the bipartite graph $G$ from Phase~\ref{step:match-others}. In $G$, consider a fractional matching $F = \{y_{ij}\}$, where $y_{ij} = \frac{1}{|\agents_2|}$ for every agent $i\in \agents_2$ and item $j \in B_i$, and $y_{ij} = 0$ if $j \not\in B_i$. Observe that for every agent $i\in \agents_2$ we have $\sum_{j\in B_2} y_{ij} = 1$ and for every item $j \in \items_2$ we have $\sum_{i \in \agents_2} y_{ij} \le \frac{1}{|\agents_2|}|\agents_2| = 1$. Hence indeed $F$ defines a fractional matching. In $\inst_2$ 
	the fractional matching $F$ gives agent $i\in \agents_2$ fractional value $\sum_{j\in B_i} y_{ij} v_i(j) = \frac{1}{|\agents_2|}v_i(B_i) = \frac{W_i}{|\agents_2|}$. 
	
	Being a fractional matching, $F$ can be represented as a distribution $D$ over integral matchings. In every one of these integral matchings, every agent $i\in \agents_2$ is matched, because $i$ is fully matched in $F$. Select a matching at random from the distribution $D$. Then in expectation, agent $i$ gets an item of value $\sum_{j\in B_i} y_{ij} v_i(j) = \frac{W_i}{|\agents_2|}$. Using $E_D$ to denote expectation over choice from distribution $D$, and denoting by $e_i$ the item received by $i$, we have that $E_D[v_i(e_i)] =  \frac{W_i}{|\agents_2|}$. Hence the expectation of $f_i$ is $E_D[\frac{\frac{v_i(\items_2)}{|\agents_2|} - v_i(e_i)}{v_i(\items_2) - W_i}] = \frac{\frac{v_i(\items_2)}{|\agents_2|} - \frac{W_i}{|\agents_2|}}{v_i(\items_2) - W_i} = \frac{1}{|\agents_2|}$. By linearity of expectation, $E_D[\sum_{i\in \agents_2} f_i] = 1$. This implies that there is a matching in $G$ under which the sum of the respective $f_i$ satisfies $\sum_{i\in \agents_2} f_i \le 1$.  The matching that maximizes $\sum_{i \in [n]} \frac{v_i(e_i)}{v_i(\items_2) - W_i}$ (which is $B^*$ that we use in the matching step, because we defined $w_i(j)$ to be $\frac{v_i(j)}{v_i(\items_2) - W_i}$) also minimizes $\sum_{i\in \agents_2} f_i$, and hence has $\sum_{i\in \agents_2} f_i \le 1$. This implies that the solution with $x_{ij} = f_i$ for every $i \in \agents_2$ and $j \in \items_3$, and $x_{ij} = 0$ for every $i \in \agents_1$, is feasible for LP3.
\end{proof}

Let $C^*$ be a fractional allocation of $\items_3$ that is an optimal solution to LP3. Phase~\ref{step:LP3} ends by performing faithful randomized rounding of $C^*$.
The following proposition follows immediately from the properties of $C^*$ and Lemma~\ref{lem:faithful}, and hence its proof is omitted.

\begin{proposition}
	\label{pro:secondFaithful}
	The faithful randomized rounding of $C^*$ produces a distribution over allocations of the items of $\items_3$, with the following properties:
	
	\begin{enumerate}
		
		\item The distribution is supported on at most $m + 1$ allocations. (The number of constraints in LP3 is $|\items_3| + n + |\agents_2|$. In a basic feasible solution, at least $|\agents_2|$ of the $s_i$ variables are positive, and so at most $|\items_3| + n = m$ of the $x_{ij}$ variables are positive.) 
		
		\item Every agent $i\in \agents_2$ gets ex-ante value $s_i \ge f_i \cdot v_i(\items_3)$.
		
		\item Every agent $i\in \agents_2$ gets ex-post value at least $s_i$, up to one item. That is, at least  $s_i - \max_{j\in \items_3}[v_i(j)]$. 
		
	\end{enumerate}
	
\end{proposition}

{The allocation algorithm above computes a distribution over $4n$ matchings in Phases~\ref{step:match-exceptional} and \ref{step:match-others}, and for each such matching, in Phase~\ref{step:LP3} it computes a distribution over $m+1$ allocations of $\items_3$. We thus have a distribution over $4n(m+1)$ allocations and we next prove that it satisfies the requirements of Theorem~\ref{thm:BoB} (except the support reduction to $n$ allocations, that will be handled in Phase~\ref{step:support} below).}

{\bf Every agent gets her proportional share ex-ante.} By item~3 of Proposition~\ref{pro:firstFaithful}, with respect to $A^*$, every agent $i$ gets value at least $PS_i(\inst)$ ex-ante. However, this value might have been attained by being allocated the respective auxiliary item $a_i$, of value $TPS_i(\inst)$. In this case, agent $i$ does not actually get $a_i$, but is instead included in $\agents_2$. Hence we need to show that for every agent $i\in \agents_2$, her combined ex-ante value from Phases~\ref{step:match-others} and~\ref{step:LP3} is at least $TPS_i(\inst)$. This ex-ante value is at least $v_i(e_i) + f_i\cdot v_i(\items_3)$. We claim that indeed $v_i(e_i) + f_i \cdot v_i(\items_3) \ge TPS_i(\inst)$.

Recall that $f_i = \frac{\frac{v_i(\items_2)}{|\agents_2|} - v_i(e_i)}{v_i(\items_2) - W_i}$. 
Observe also that $v_i(\items_2) - W_i \le v_i(\items_3)$,  because the total value for $i$ of the $|\agents_2|$ items allocated under $B^*$ cannot be larger than $W_i=v_i(B_i)$ (as $B_i$ contains the $|\agents_2|$ items of highest value). Combining these observations we have that:

$$f_i =  \frac{\frac{v_i(\items_2)}{|\agents_2|} - v_i(e_i)}{v_i(\items_2) - W_i} \ge \frac{\frac{v_i(\items_2)}{|\agents_2|} - v_i(e_i)}{v_i(\items_3)} =  \frac{v_i(\items_2) - |\agents_2| \cdot v_i(e_i)}{|\agents_2| \cdot v_i(\items_3)}$$

We can now establish the claim.

$$v_i(e_i) + f_i v_i(\items_3) \ge v_i(e_i) + \frac{v_i(\items_2) - |\agents_2|\cdot v_i(e_i)}{|\agents_2| \cdot v_i(\items_3)} v_i(\items_3) = \frac{v_i(\items_2)}{|\agents_2|} \ge TPS_i(\inst)$$
(for the last equality, see discussion in Phase~\ref{step:match-others}).
\medskip

{\bf Every agent gets at least half her TPS ex-post.} For agents in $\agents_1$, this holds by definition. For agents $i\in \agents_2$, we have already shown that ex-ante they get at least $TPS_i(\inst)$. Item~3 of Proposition~\ref{pro:secondFaithful} implies that ex-post agent $i$ gets a value of at least $TPS_i(\inst) - \max_{j\in \items_3}[v_i(j)]$. If $\max_{j\in \items_3}[v_i(j)] \le \frac{TPS_i(\inst)}{2}$, then at least a value of $\frac{TPS_i(\inst)}{2}$ remains. If $\max_{j\in \items_3}[v_i(j)] > \frac{TPS_i(\inst)}{2}$, then also $v_i(e_i) \geq \frac{TPS_i(\inst)}{2}$ (by Proposition~\ref{pro:firstItem}), and hence $i$ gets half her TPS already after Phase~\ref{step:match-others}.    
\medskip

{\bf The allocation is prop1 ex-post.} {If there is an item that is exceptional for agent $i$, then  an item that $i$ values most, denoted as item $j$, necessarily satisfies $v_i(j) \ge PS_i$ (if $v_i(j) < PS_i$ then $TPS_i = PS_i$, and then $j$ is not exceptional for $i$). In this case, every allocation gives $i$ her proportional share, up to the item $j$.} If there is no item that is exceptional for agent $i$, then $TPS_i(\inst) = PS_i(\inst)$, and also, $i$ ends up in $\agents_2$. Item~3 of Proposition~\ref{pro:secondFaithful} ensures that she gets $TPS_i(\inst)$ up to one item, which in this case is equivalent to $PS_i(\inst)$ up to one item.
\medskip

{\bf The randomized allocation is computed in polynomial time.} The TPS of every agent can be computed in polynomial time. In various steps, the algorithm involves scaling of the valuation functions by constant factors, which too can be done in polynomial time. The heavier computational aspects of the algorithm are the following. Phases~\ref{step:match-exceptional} and~\ref{step:LP3} each involve solving an LP, and then performing faithful randomized rounding. Phase~\ref{step:match-others} involves finding a maximum weight matching. Also, these heavier steps can be done in polynomial time, using standard algorithms.
\medskip

{\bf The randomized allocation is supported on $n$ allocations.} The combination of item~2 of Proposition~\ref{pro:firstFaithful} and item~1 of Proposition~\ref{pro:secondFaithful} implies that the randomized allocation is supported over at most $4n(m+1)$ allocations. In Phase~\ref{step:support} (to be described next) of our algorithm, we reduce this number to $n$.
\medskip

{\bf Phase~\ref{step:support}:} Reducing the support size to be at most $n$.

So far we established that there is a distribution $D$ over allocations, giving every agent at least her proportional share ex-ante, and supported on at most $4n(m+1)$ ``good" allocations: allocations that give every agent at least half her TPS, and are Prop1. We now explain how to reduce the support size to at most $n$. 

Set up the following linear program. For every allocation in the support of $D$ (index these allocations as $G^k$) there is a variable $z_k$ (representing the probability that $G^k$ is selected in our new distribution), and for every agent $i$ there is a variable $y_i$ (representing her ex-ante value). 
The coefficients $a_{ik}$ denote the value that agent $i$ derives from the items allocated to him under allocation $G^k$. 
The linear program LP4 is as follows:
\medskip

{\bf Minimize $z$ subject to:}

\begin{enumerate}

	\item $\sum_k z_k = z$. 
	
	\item $y_i = \sum_k a_{ik} z_k$ for every agent $i\in \agents$. ($y_i$ represents the ex-ante value of the randomized allocation to agent $i$.)
	
	\item $y_i \ge PS_i$ for every $i\in \agents$. (Every agent gets at least her proportional share ex-ante.)
	
	\item $z_k \ge 0$ for every $G^k$ ($z_k$ is proportional to the probability of $G^k$). 
	
\end{enumerate}

The distribution $D$ shows that the optimal value $z^*$ of LP4 satisfies $z^* \le 1$. LP4 has $2n+1$ constraints (excluding non-negativity constraints), and hence an optimal basic feasible solution is supported on at most $2n+1$ positive variables. As $z$ and the variables $y_i$ are all positive, there are at most $n$ variables $z_k$ that are positive. Scale the $z_k$ variables of the optimal solution by $\frac{1}{z^*}$, so that they form a probability distribution. Likewise, scale the $y_i$ variables by $\frac{1}{z^*}$ so that constraint~2 remains satisfied. Constraint~3 also remains satisfied, as $\frac{1}{z^*}\geq 1$. The scaled values of the $z_k$ variables represent a randomized allocation that proves Theorem~\ref{thm:BoB}.

	\section{Discussion}
		\label{sec:discussion}

We have presented a best-of-both-worlds result, showing that for every allocation instance with additive valuations, there is a randomized allocation that gives every agent at least her proportional share ex-ante, and at least half of her TPS (and MMS) ex-post. 
Moreover, we have shown that there is a deterministic polynomial time algorithm that, given the valuation functions of the agents, computes a faithful 
implementation of such a randomized allocation, supported on at most $n$ allocations.
We next discuss directions in which our results can possibly be improved upon, presenting impossibilities of some natural extensions, as well as some open problems.

\subsection{Other Fairness Guarantees}


Theorem~\ref{thm:BoB} guarantees every agent at least half her TPS ex-post. This is nearly the best possible, in the sense that there are allocation instances for which no allocation gives every agent more than a $\frac{n}{2n-1}$ fraction of her TPS. Still, it might be interesting to see if the ex-post BoBW guarantee can be improved to $\frac{n}{2n-1}$ of the TPS, to match the lower bound.

For the maximin share, MMS, it might be possible to offer every agent a $\rho$-fraction of her MMS ex-post, for some $\frac{1}{2} < \rho < 1$ significantly larger than $\frac{1}{2}$. Even if so, it is not clear if such a guarantee will be better than half the TPS, because the gap between MMS and TPS may be a factor of $2 - \frac{1}{n}$.

Our focus in this paper is on share-based fairness notions (such as proportional, MMS and TPS). Other works (see Theorems~\ref{thm:Aziz20} and~\ref{thm:FSV20}) addressed envy-based notions (such as EF and EF1). We briefly discuss here whether envy-based fairness notions offer agents higher value than share-based notions. Our allocation is ex-ante proportional, but not necessarily ex-ante envy free. One might argue that guaranteeing an ex-ante envy-free allocation (as indeed achieved in Theorems~\ref{thm:Aziz20} and~\ref{thm:FSV20}) offers agents higher value ex-ante, as every envy-free allocation is also proportional, but there are proportional allocations that are not envy free. However, as every fractional proportional allocation Pareto dominates a fractional envy free allocation (the trivial allocation in which every agent gets a $\frac{1}{n}$ fraction of every item), an ex-ante envy freeness guarantee by itself offers no advantage over an ex-ante proportional guarantee, in terms of the value that it guarantees to agents. 

As to ex-post guarantees, Example~\ref{ex:notMMS} (among others) illustrates that EF1 allocations might give agents value that is a factor of $\Omega(n)$ smaller than their TPS (and MMS). Our allocations guarantee every agent at least half her TPS ex-post. In this paper we did not aim to also get EF1. 
A direction for future work is to obtain EF1 on top of the properties we obtain (TPS approximation and ex-ante proportionality).
Our algorithm does not obtain all these properties (it is not EF1).
In Appendix \ref{app:EF1} we present a result that shows that EF1 is not in conflict with the combination of approximate TPS and approximate proportionality (yet leave open the question for exact proportionality). Specifically, we present a randomized allocation that does achieve (ex-post) both EF1 as well as an approximate TPS guarantee, but only an approximate proportional guarantee ex-ante (concretely, it achieves $\frac{n}{2n-1}$-TPS and EF1 ex-post, and $\frac{n}{2n-1}$-proportional share ex-ante).


An alternative BoBW result that one might consider is a result in which we replace the ex-post guarantee to be EFX. 
As every EFX allocation gives every agent a $\frac{n}{2n-1}$ fraction of her TPS {(see proof in Appendix~\ref{app:proofs})}, such a result will, in particular, strengthen our result and obtain the best possible TPS fraction of $\frac{n}{2n-1}$. 
{For two agents, running the standard cut-and-choose protocol with a random cutter gives a randomized allocation that is proportional ex-ante (which for two agents implies also EF ex-ante), and both EFX and MMS ex-post (and $2/3$-TPS).}
A major hurdle in obtaining such a {combination of properties for arbitrarily many agents} is that  EFX allocations are not known to always exist beyond 3 agents, so achieving EFX in the BoBW setting seems to be currently beyond reach. 
Moreover, \citet{plaut2020almost} showed that EFX conflicts with Pareto optimality\footnote{Their example for the additive case uses goods with zero values. Such examples are not known for instances with no zero-value items.}, another important property that we wish to have, thus we cannot have an EFX-based BoBW result with Pareto optimality. 

\subsection{Economic Efficiency}

As noted in Corollary~\ref{cor:BoB}, our BoBW result can be supported on allocations that are Pareto optimal  (though our polynomial time randomized allocation of Theorem~\ref{thm:BoB} does not guarantee ex-post PO). Some stronger economic efficiency properties cannot be achieved, as they contradict our ex-post fairness properties. For example, ex-ante Pareto optimal (fPO) cannot be achieved, as shown in Proposition~\ref{pro:noParetoExAnte}. Likewise,
the next example shows that one should not attempt to approximately maximize welfare, not even if valuation functions are normalized and scaled so that every agent has the same value for the set of all items $\items$. 

Consider the following example with $m=n$, with $n$ being a perfect square. For every agent $i \le \sqrt{n}$, every item $j$ with $j = i$ modulo $\sqrt{n}$ has value $\sqrt{n}$ (and the rest of the items have value~0). For the remaining agents, every item has value~1. Observe that $v_i(\items) = n$ for every agent $i$, so valuations are indeed normalized.
The TPS of each of the remaining agents is~1, and hence if we want every agent to get at least a constant fraction of her TPS, 
in every ex-post allocation each of them must receive at least one item. The welfare of every such allocation is at most $n - \sqrt{n} + \sqrt{n} \cdot \sqrt{n} < 2n$, whereas the maximum welfare allocation gives each of the first $\sqrt{n}$ agents value of $n$ ($\sqrt{n}$ items, each of value $\sqrt{n}$), resulting in optimal welfare of $n^{\frac{3}{2}}$. Thus, any allocation that gives every agent a constant fraction of her TPS does not approximate the maximum welfare (even when valuations are normalized) to any  factor better than $\Omega(\sqrt{n})$.

This leads us to consider Nash Social Welfare (NSW). We cannot hope to exactly maximize the fractional NSW (fNSW) ex-ante, as this allocation is fractionally PO, and Proposition~\ref{pro:noParetoExAnte} shows an impossibility result for this case. However, we can hope to get a constant approximation for the maximum fNSW ex-ante. 
It would be interesting to understand whether some version of our allocation algorithm from the proof of Theorem~\ref{thm:BoB} (with suitable modifications) could provide such a result. 
{ In Appendix \ref{app:fNSW} we show that in the simplest non-trivial case (when there are $n$ items to allocate to $n$ additive agents) there is a randomized allocation algorithm that obtains fNSW approximation ex-ante, while giving every agent at least her proportional share ex-ante and at least her TPS ex-post}. This demonstrates that at least in simple settings, it is possible to achieve fNSW approximation ex-ante together with the other properties we are after.

\subsection{Incentive Compatibility}
\label{sec:truthful}

We discuss here truthfulness aspects for individual agents, and do not address in our discussion more demanding aspects of group strategyproofness.

BoBW allocation mechanisms are randomized. As such, one may consider either ex-post or ex-ante truthfulness notions. The most straightforward notion is that of {\em universal truthfulness} -- reporting the true valuation function is a dominant strategy, with respect to both ex-post and ex-ante values simultaneously, for every realization of the coin tosses of the randomized allocation. 
A BoBW result with universal truthfulness was achieved in~\citep{BEF20} in the special case of additive {\em dichotomous} valuations (and also for submodular dichotomous valuations). However, for general additive valuations, there are impossibility results for truthful mechanisms, and they carry over to universally truthful mechanisms. In particular, it is proved in~\citep{ABCM2017} that every truthful allocation mechanism for two agents that allocates all items must, in some instances,  give an agent no more than a $\frac{2}{m}$ fraction of her MMS. Consequently, every universally truthful randomized allocation mechanism for two agents that allocates all items must sometimes give an agent no more than a $\frac{2}{m}$ fraction of her MMS ex-post. 
Moreover, this implies the next proposition regarding $n$ agents. The proposition is  a direct corollary from the result of \citet{ABCM2017}.  
\begin{proposition}
	Every universally truthful randomized allocation mechanism for $n$ agents and $m$ items that is ex-post PO must sometimes not give an agent more than an $O(\frac{n}{m})$ fraction of her MMS ex-post.
\end{proposition}
This proposition holds by adapting the two agents impossibility result proof, having large enough $m$ and adding one auxiliary item and $n-2$ auxiliary agents that each wants only the auxiliary item.  

An alternative truthfulness notion is that of {\em truthful in expectation} (TIE), which is an ex-ante notion, and postulates that agents attempt to maximize their expected utility. In the BoBW setting, it is indeed reasonable to assume that agents are expectation maximizers, as these are the type of guarantees that they are given ex-ante.
However, TIE tacitly assumes that agents are not strategic concerning their ex-post guarantees, an aspect that is somewhat problematic in BoBW settings. A TIE BoBW result for additive dichotomous valuations is achieved in~\citep{HPPS20}.

None of the BoBW results (the previous Theorems~\ref{thm:FSV20} and~\ref{thm:Aziz20}, and our Theorem~\ref{thm:BoB}) provides a truthful mechanism, not even TIE. On the other hand, the weaker and rather trivial Proposition~\ref{pro:uniform} does give a TIE mechanism, but with rather weak BoBW guarantees. 

In the mechanism of Proposition~\ref{pro:uniform} an agent that maximizes expected utility has no incentive to lie, but also no incentive to be truthful. Using a trick of~\citet{MT10}, we can modify that mechanism so as to make truthfulness the unique dominant strategy. After receiving the valuation functions of all agents, first generate a fractional solution $A^f_R$ at random. If $A^f_R$ Pareto dominates the uniform fractional allocation, then faithfully implement $A^f_R$. If not, then faithfully implement the uniform fractional allocation.  This mechanism is ex-ante proportional and ex-post Prop1.

We do not know if there is a TIE mechanism that offers every agent at least a constant fraction of her MMS ex-post (even without requiring any ex-ante guarantee).


\subsection{Extension to Chores (Bads)}

We briefly discuss here the BoBW setting for indivisible chores. As we shall see, known results for deterministic allocation of chores easily lead to BoBW results for chores that are comparable to, or even stronger than, the ones that we prove for goods.

Recall the setting of allocating a set $\items$ of indivisible items to $n$ agents with additive valuations. Items are referred to as {\em chores} (or {\em bads}) if $v_i(j) \le 0$ for every agent $i$ and item $j\in \items$. {Thus, this chore has a \emph{cost} of $-v_i(j)$.}
In this setting, agents prefer not to receive any item, but {all chores must be taken care of, so the allocation is required to allocate all of them.} 

The definitions of proportional share (PS) and maximin share (MMS) with respect to chores are exactly the same as they are with respect to goods. For the truncated proportional share (TPS), we propose the following simple definition with respect to chores: $TPS_i = \min[PS_i, \min_{j\in \items}[v_i(j)]]$. In analogy with the definition of TPS with respect to goods, if there are no exceptional items, then the TPS for chores is equal to the proportional share. For chores, an item is exceptional if its {\em cost} is larger than the cost of the proportional share (or equivalently, its value is more negative than the proportional share), whereas for goods an item is exceptional if its value is larger than the value of the proportional share. In the presence of exceptional items, the definition of TPS for goods needs to account for the possibility that an agent other that our agent receives the good, whereas the definition of TPS for chores needs to account for the possibility that our agent receives the chore. As these are different types of events, this leads to technical differences between the two definitions of TPS, the one for goods and the one for chores. 

For chores, the value of 
the proportional share, the MMS and the TPS are all
negative, unless $v_i$ is identically~0. Note also that as in the case for goods, also for chores the following inequalities hold (for additive valuations): $PS_i \ge TPS_i \ge MMS_i$. 

In~\citep{ARSW17} it is shown (adapting results of~\cite{KurokawaPW18} from goods to chores) that there are allocation instances for indivisible chores for agents with additive valuations in which no allocation gives every agent value at least her MMS. On the positive side, it was shown that letting agents choose chores in a round-robin fashion assigns every agent chores of cost at most  $2 - \frac{1}{n}$ times the cost of her MMS. 
This is a  $(2 - \frac{1}{n})$-{approximation} 
of her MMS. (Note that as the MMS is negative, approximation ratios are at least~1 rather than at most~1.) In fact, their proof shows that every agent $i$ gets value at least $PS_i + (1 - \frac{1}{n})\min_{j\in \items}[v_i(j)]$, and hence {no worse than} a $(2 - \frac{1}{n})$-{approximation} of her TPS. 
For TPS, this is nearly the best possible guarantee. Consider an instance with $n+1$ items each of value $-1$. The TPS of every agent is $-\frac{n+1}{n}$, whereas in every allocation, some agent gets a bundle of value not better than $-2$. Hence in every allocation, some agent gets a bundle of value no better than $2 - \frac{2}{n+1}$ {times} her TPS.


The round-robin allocation can easily be transformed into a BoBW result by using a random permutation to determine the order among agents. In this case, every agent gets at least her proportional share ex-ante, and chores of cost at most  $2 - \frac{1}{n}$ times her TPS cost ex-post. Moreover, this allocation is also EF1 ex-post.

Unlike the case of goods, where the {\em up-to-one-good} paradigm does not offer any approximation for the TPS, for chores, the up-to-one-good paradigm guarantees a factor~2 approximation of the TPS. Consequently, approaches similar to those of~\citep{FSV20} (with straightforward adaptation to chores instead of goods) can be used in order to get fractional allocations that are fractionally Pareto optimal both ex-ante and ex-post, are proportional ex-ante, and {assign every agent chores of cost at most twice her TPS cost ex-post.}

If one considers ex-post guarantees with respect to the MMS instead of the TPS, better approximation ratios can be achieved ex-post. In~\citep{ALW20} it was shown that a certain picking order leads to allocations that 
assign 
every agent chores of cost {no worse than} 
$\frac{5}{3}$ times the cost of 
her MMS. 
Also this allocation can be transformed into a BoBW result by using a random permutation to determine the order among agents. 

\begin{observation}
The allocation mechanism of \citep{ALW20} with a random permutation over the agents assigns every agent chores of cost at most $\frac{5}{3}$  times her MMS cost ex-post, and in  expectation her cost is at most her proportional share. 
\end{observation}

For deterministic allocations (not as part of a BoBW result), the $\frac{5}{3}$ ratio is known not to be best possible. There are allocations that assign 
every agent chores of cost at most $\frac{11}{9}$ times 
the cost of her MMS~\citep{HL19}.

For allocation instances that involve a mixture of goods and chores, known also as {\em mixed manna}, there are instances in which agents have additive valuations, the MMS of every agent is strictly positive, whereas in every allocation (that allocates all items) some agent receives a bundle of value at most~0 \citep{KMT2020}. 
Hence, for mixed manna, it is not always possible to find an allocation that gives every agent a positive fraction of her MMS.

	
	

	\bibliographystyle{abbrvnat}
	\bibliography{bib}

\begin{thebibliography}{31}
\providecommand{\natexlab}[1]{#1}
\providecommand{\url}[1]{\texttt{#1}}
\expandafter\ifx\csname urlstyle\endcsname\relax
  \providecommand{\doi}[1]{doi: #1}\else
  \providecommand{\doi}{doi: \begingroup \urlstyle{rm}\Url}\fi

\bibitem[Amanatidis et~al.(2017{\natexlab{a}})Amanatidis, Birmpas,
  Christodoulou, and Markakis]{ABCM2017}
G.~Amanatidis, G.~Birmpas, G.~Christodoulou, and E.~Markakis.
\newblock Truthful allocation mechanisms without payments: Characterization and
  implications on fairness.
\newblock In \emph{Proceedings of the 2017 ACM Conference on Economics and
  Computation}, EC'17, pages 545--562, 2017{\natexlab{a}}.

\bibitem[Amanatidis et~al.(2017{\natexlab{b}})Amanatidis, Markakis, Nikzad, and
  Saberi]{amanatidis2017approximation}
G.~Amanatidis, E.~Markakis, A.~Nikzad, and A.~Saberi.
\newblock Approximation algorithms for computing maximin share allocations.
\newblock \emph{ACM Transactions on Algorithms (TALG)}, 13\penalty0
  (4):\penalty0 1--28, 2017{\natexlab{b}}.

\bibitem[Aziz(2020)]{Aziz20}
H.~Aziz.
\newblock Simultaneously achieving ex-ante and ex-post fairness.
\newblock In \emph{International Conference on Web and Internet Economics},
  pages 341--355. Springer, 2020.

\bibitem[Aziz et~al.(2016)Aziz, Bir{\'{o}}, Lang, Lesca, and Monnot]{ABLLM16}
H.~Aziz, P.~Bir{\'{o}}, J.~Lang, J.~Lesca, and J.~Monnot.
\newblock Optimal reallocation under additive and ordinal preferences.
\newblock In \emph{Proceedings of the 2016 International Conference on
  Autonomous Agents {\&} Multiagent Systems, Singapore, May 9-13, 2016}, pages
  402--410. {ACM}, 2016.

\bibitem[Aziz et~al.(2017)Aziz, Rauchecker, Schryen, and Walsh]{ARSW17}
H.~Aziz, G.~Rauchecker, G.~Schryen, and T.~Walsh.
\newblock Algorithms for max-min share fair allocation of indivisible chores.
\newblock In \emph{Proceedings of the Thirty-First {AAAI} Conference on
  Artificial Intelligence, February 4-9, 2017, San Francisco, California,
  {USA}}, pages 335--341. {AAAI} Press, 2017.

\bibitem[Aziz et~al.(2020)Aziz, Li, and Wu]{ALW20}
H.~Aziz, B.~Li, and X.~Wu.
\newblock Approximate and strategyproof maximin share allocation of chores with
  ordinal preferences.
\newblock \emph{CoRR}, abs/2012.13884, 2020.

\bibitem[Babaioff et~al.(2020)Babaioff, Ezra, and Feige]{BEF20}
M.~Babaioff, T.~Ezra, and U.~Feige.
\newblock Fair and truthful mechanisms for dichotomous valuations.
\newblock \emph{arXiv preprint arXiv:2002.10704}, 2020.

\bibitem[Babaioff et~al.(2021)Babaioff, Ezra, and Feige]{babaioff2021fair}
M.~Babaioff, T.~Ezra, and U.~Feige.
\newblock Fair-share allocations for agents with arbitrary entitlements.
\newblock \emph{arXiv preprint arXiv:2103.04304}, 2021.
\newblock To appear in Proceedings of the 2021 {ACM} Conference on Economics
  and Computation, {ACM-EC} '21.

\bibitem[Barman and Krishnamurthy(2019)]{BarmanK19}
S.~Barman and S.~K. Krishnamurthy.
\newblock On the proximity of markets with integral equilibria.
\newblock In \emph{The Thirty-Third {AAAI} Conference on Artificial
  Intelligence, {AAAI}}, pages 1748--1755, 2019.

\bibitem[Barman and Krishnamurthy(2020)]{BK20}
S.~Barman and S.~K. Krishnamurthy.
\newblock Approximation algorithms for maximin fair division.
\newblock \emph{ACM Transactions on Economics and Computation (TEAC)},
  8\penalty0 (1):\penalty0 1--28, 2020.

\bibitem[Bez{\'a}kov{\'a} and Dani(2005)]{BD05}
I.~Bez{\'a}kov{\'a} and V.~Dani.
\newblock Allocating indivisible goods.
\newblock \emph{ACM SIGecom Exchanges}, 5\penalty0 (3):\penalty0 11--18, 2005.

\bibitem[Bogomolnaia and Moulin(2001)]{BM01}
A.~Bogomolnaia and H.~Moulin.
\newblock A new solution to the random assignment problem.
\newblock \emph{Journal of Economic theory}, 100\penalty0 (2):\penalty0
  295--328, 2001.

\bibitem[Budish(2011)]{Budish11}
E.~Budish.
\newblock The combinatorial assignment problem: Approximate competitive
  equilibrium from equal incomes.
\newblock \emph{Journal of Political Economy}, 119\penalty0 (6):\penalty0
  1061--1103, 2011.

\bibitem[Budish et~al.(2013)Budish, Che, Kojima, and Milgrom]{BCKM13}
E.~Budish, Y.-K. Che, F.~Kojima, and P.~Milgrom.
\newblock Designing random allocation mechanisms: Theory and applications.
\newblock \emph{American economic review}, 103\penalty0 (2):\penalty0 585--623,
  2013.

\bibitem[Caragiannis et~al.(2019)Caragiannis, Kurokawa, Moulin, Procaccia,
  Shah, and Wang]{CKMPSW19}
I.~Caragiannis, D.~Kurokawa, H.~Moulin, A.~D. Procaccia, N.~Shah, and J.~Wang.
\newblock The unreasonable fairness of maximum nash welfare.
\newblock \emph{ACM Transactions on Economics and Computation (TEAC)},
  7\penalty0 (3):\penalty0 1--32, 2019.

\bibitem[Conitzer et~al.(2017)Conitzer, Freeman, and Shah]{conitzer2017fair}
V.~Conitzer, R.~Freeman, and N.~Shah.
\newblock Fair public decision making.
\newblock In \emph{Proceedings of the 2017 ACM Conference on Economics and
  Computation}, pages 629--646, 2017.

\bibitem[De~Keijzer et~al.(2009)De~Keijzer, Bouveret, Klos, and Zhang]{dKBKZ09}
B.~De~Keijzer, S.~Bouveret, T.~Klos, and Y.~Zhang.
\newblock On the complexity of efficiency and envy-freeness in fair division of
  indivisible goods with additive preferences.
\newblock In \emph{International Conference on Algorithmic DecisionTheory},
  pages 98--110. Springer, 2009.

\bibitem[Feige et~al.(2021)Feige, Sapir, and Tauber]{FST21}
U.~Feige, A.~Sapir, and L.~Tauber.
\newblock A tight negative example for {MMS} fair allocations.
\newblock \emph{CoRR}, abs/2104.04977, 2021.

\bibitem[Freeman et~al.(2020)Freeman, Shah, and Vaish]{FSV20}
R.~Freeman, N.~Shah, and R.~Vaish.
\newblock Best of both worlds: ex-ante and ex-post fairness in resource
  allocation.
\newblock In \emph{Proceedings of the 21st ACM Conference on Economics and
  Computation}, pages 21--22, 2020.

\bibitem[Garg and Taki(2020)]{GT20}
J.~Garg and S.~Taki.
\newblock An improved approximation algorithm for maximin shares.
\newblock In \emph{Proceedings of the 21st ACM Conference on Economics and
  Computation}, pages 379--380, 2020.

\bibitem[Garg et~al.(2019)Garg, McGlaughlin, and Taki]{garg2019approximating}
J.~Garg, P.~McGlaughlin, and S.~Taki.
\newblock Approximating maximin share allocations.
\newblock \emph{Open access series in informatics}, 69, 2019.

\bibitem[Ghodsi et~al.(2018)Ghodsi, Hajiaghayi, Seddighin, Seddighin, and
  Yami]{GhodsiHSSY18}
M.~Ghodsi, M.~T. Hajiaghayi, M.~Seddighin, S.~Seddighin, and H.~Yami.
\newblock Fair allocation of indivisible goods: Improvements and
  generalizations.
\newblock In \emph{Proceedings of the 2018 {ACM} Conference on Economics and
  Computation, Ithaca, NY, USA, June 18-22, 2018}, pages 539--556. {ACM}, 2018.

\bibitem[Halpern et~al.(2020)Halpern, Procaccia, Psomas, and Shah]{HPPS20}
D.~Halpern, A.~D. Procaccia, A.~Psomas, and N.~Shah.
\newblock Fair division with binary valuations: One rule to rule them all.
\newblock In \emph{Web and Internet Economics - 16th International Conference},
  2020.

\bibitem[Huang and Lu(2019)]{HL19}
X.~Huang and P.~Lu.
\newblock An algorithmic framework for approximating maximin share allocation
  of chores.
\newblock \emph{CoRR}, abs/1907.04505, 2019.

\bibitem[Kulkarni et~al.(2020)Kulkarni, Mehta, and Taki]{KMT2020}
R.~Kulkarni, R.~Mehta, and S.~Taki.
\newblock Approximating maximin shares with mixed manna.
\newblock \emph{CoRR}, abs/2007.09133, 2020.

\bibitem[Kurokawa et~al.(2018)Kurokawa, Procaccia, and Wang]{KurokawaPW18}
D.~Kurokawa, A.~D. Procaccia, and J.~Wang.
\newblock Fair enough: Guaranteeing approximate maximin shares.
\newblock \emph{J. {ACM}}, 65\penalty0 (2):\penalty0 8:1--8:27, 2018.

\bibitem[Lenstra et~al.(1990)Lenstra, Shmoys, and Tardos]{LST90}
J.~K. Lenstra, D.~B. Shmoys, and {\'E}.~Tardos.
\newblock Approximation algorithms for scheduling unrelated parallel machines.
\newblock \emph{Mathematical programming}, 46\penalty0 (1):\penalty0 259--271,
  1990.

\bibitem[Lipton et~al.(2004)Lipton, Markakis, Mossel, and Saberi]{LMMS04}
R.~J. Lipton, E.~Markakis, E.~Mossel, and A.~Saberi.
\newblock On approximately fair allocations of indivisible goods.
\newblock In \emph{Proceedings of the 5th {ACM} Conference on Electronic
  Commerce}, EC'04, pages 125--131, 2004.

\bibitem[Mossel and Tamuz(2010)]{MT10}
E.~Mossel and O.~Tamuz.
\newblock Truthful fair division.
\newblock In \emph{Proceeding of the 3rd International Symposium on Algorithmic
  Game Theory}, SAGT'10, pages 288--299, 2010.

\bibitem[Plaut and Roughgarden(2020)]{plaut2020almost}
B.~Plaut and T.~Roughgarden.
\newblock Almost envy-freeness with general valuations.
\newblock \emph{SIAM Journal on Discrete Mathematics}, 34\penalty0
  (2):\penalty0 1039--1068, 2020.

\bibitem[Srinivasan(2008)]{Srinivasan08}
A.~Srinivasan.
\newblock Budgeted allocations in the full-information setting.
\newblock In \emph{Approximation, Randomization and Combinatorial Optimization.
  Algorithms and Techniques}, pages 247--253. Springer, 2008.

\end{thebibliography}

	\appendix
	\section{Missing Proofs}
\label{app:proofs}

We first show that there is a polynomial time algorithm that gives every agent (with an additive valuation) at least a $\frac{n}{2n-1}$ fraction of her TPS. 

We say that an allocation $A = (A_1, \ldots, A_n)$ is {\em half-fair} if for every two agents $i$ and $j$, if $|A_j| > 1$ then $v_i(A_i) \ge \frac{1}{2}v_i(A_j)$. In other words, if an agent $i$ prefers bundle $A_j$ over her own bundle $A_i$, then either $A_j$ contains only one item, or $A_j$ is at most twice as valuable to $i$ than $A_i$.

Proposition \ref{prop:halffair} below shows that in every half-fair allocation, every agent gets at least a $\frac{n}{2n-1}$ fraction of her TPS.
We note that every EFX allocation is half-fair, and thus gives every agent at least a $\frac{n}{2n-1}$ fraction of her TPS. 
Unfortunately, EFX allocations are not known to always exist. However, EF1 allocations do always exist, as shown in~\citep{LMMS04}. Though EF1 allocations 
are not necessarily
half-fair (recall Example~\ref{ex:notMMS}), the EF1 allocations generated by the algorithm of~\citet{LMMS04} are half-fair. Hence the algorithm of~\citet{LMMS04} produces an allocation that gives every agent at least a $\frac{n}{2n-1}$ fraction of her TPS.

\begin{proposition}\label{prop:halffair}
Every half-fair allocation $A = (A_1, \ldots, A_n)$ gives agent $i$ at least a $\frac{n}{2n-1}$ fraction of her TPS.
\end{proposition} 

\begin{proof}
Let $A = (A_1, \ldots, A_n)$ be a half-fair allocation. 
Recall that $TPS_i$ has the property that for every single item $e$, $TPS_i(n-1,\items \setminus \{e\}, v_i) \ge TPS_i(n, \items, v_i)$. Let $K$ denote the set of items that are in bundles (excluding $A_i$) that contain only a single item, and let $\items' = \items \setminus K$. Then $TPS_i(n-|K|,\items', v_i) \ge TPS_i(n, \items, v_i)$. As allocation $A$ is half-fair, we have that $v_i(\items')  \le (2(n-|K|) - 1)v_i(A_i)$, since the items in $\items'\setminus A_i$ are divided between  $n-|K|-1$ agents, and each of those agents has a value (according to $v_i$) of at most $2v_i( A_i)$. Hence:

$$TPS_i \le TPS_i(n-|K|,\items', v_i) \le PS_i(n-|K|,\items', v_i) \le \frac{2(n-|K|) - 1}{n - |K|}v_i(A_i) \le \frac{2n - 1}{n}v_i(A_i),$$
which concludes the proof.
\end{proof}


We now restate and prove Proposition~\ref{pro:noParetoExAnte}, showing that ex-ante fPO is in conflict with the ex-post fairness properties that we desire.
\propEasyBoBW*
\begin{proof}
Recall the instance of Example~\ref{ex:notMMS}, where the maximin share of every agent is $n$. Consider an arbitrary fractional fPO allocation $A^*$ for this instance. We may assume that every agent $i$ receives fractions either from at least one of the items  $\{b_1,b_2, \ldots , b_{n-1}\} $ 
or from at least two of the items in $\{s_1, \ldots, s_n\}$, as otherwise agent $i$ cannot get ex-post a value larger than $1 + \epsilon$. In either of these cases, agent $i$ holds some fraction of an item (say, item $e_j$) different from $s_i$. Fractional Pareto optimality of $A^*$ then implies that $A^*$ allocates $s_i$ in full to agent $i$. (Otherwise $A^*$ can be Pareto improved. Agent $i$, who values $s_i$ more than other agents do, and values $e_j$ not more than other agents do, can trade a fraction of $e_j$ with a fraction of $s_i$, benefiting himself, and without hurting the agent who originally holds the fraction of item $s_i$.)
Consequently, every agent $i$ receives the corresponding item $s_i$ in every ex-post allocation.  By the pigeon-hole principle, in an ex-post allocation there is an agent that receives no item among  $\{b_1,b_2, \ldots , b_{n-1}\} $. This agent $i$ receives only $s_i$, and hence only a $\frac{1 + \epsilon}{n}$ fraction of her MMS.
\end{proof}

	\section{Faithful Implementations of Fractional Allocations}
		\label{sec:faithful}


In this section we present a self contained explanation of a usage of randomized rounding to obtain BoBW fairness,  the  concept we refer to as faithful implementation. We provide some historical context as to the development of various components of it, and present the proof of Lemma \ref{lem:faithful}, which summarizes the result regarding faithful implementations.

Consider a fractional allocation $A^*$ of $m$ items to $n$ agents with additive valuations. Denote the fractional allocation to agent $i$ by $A^*_i$, {with $A^*_{ij}$ denoting the fraction of item $j$ given to agent $i$ in $A^*$. Let $M^f_i= \{j \; | \; 0<A^*_{ij}<1\}$} denote the set of items for which some positive {proper} fraction (neither~0 nor~1) is allocated to $i$, and let $f = \sum_{i \in \agents} |M^f_i|$ denote the number of variables 
that are strictly fractional.

We consider generating a distribution over integral allocations from the fractional allocation $A^*$ (a ``rounding procedure").
We distinguish between three kinds of rounding:

\begin{itemize}

\item {\em Deterministic rounding.} Produces a single integral allocation.

\item {\em Randomized rounding.} Produces a distribution over integral allocations. 

\item {\em Implementation.} Randomized rounding, where the expectation of the associated distribution is exactly $A^*$. 


\end{itemize}

We consider two notions of polynomial-time algorithms for performing randomized rounding.

\begin{itemize}

\item {\em Randomized polynomial time.} There is a randomized polynomial time algorithm that samples an integer allocation from the associated distribution.

\item {\em Deterministic polynomial time.} There is a deterministic polynomial time algorithm that lists all integral allocations in the support of the distribution, together with the associated probability of each allocation. In particular, this implies that the size of the support is upper bounded by some polynomial in $n$ and $m$.

\end{itemize}

We list several {\em faithfulness} properties that may be associated with the rounding.

\begin{enumerate}

\item {\em Ex-post faithfulness}, which satisfy both of the following properties:

\begin{enumerate}

\item {\em Faithfulness from above.} In the rounded integral allocation $A$, every agent $i$ gets a bundle of value at most her fractional value,  {up to the value of one of her fractionally allocated items.} 
That is, $v_i(A_i) \le v_i(A^*_i) + \max_{j\in \items^f_i} v_i(j)$. 

\item {\em Faithfulness from below.} In the rounded integral allocation $A$, every agent $i$ gets a bundle of value at least her fractional value, {up to the value of one of her fractionally allocated items.} 
That is, $v_i(A_i) \ge v_i(A^*_i) - \max_{j\in \items^f_i} v_i(j)$.

\end{enumerate}

For an implementation of a fractional allocation, Ex-post faithfulness follows from the following single property: 
\begin{itemize}

\item {\em Small spread.} For every agent $i$, the difference in values that $i$ receives in any two rounded integral allocations is at most $\max_{j\in \items^f_i} v_i(j)$. 

\end{itemize}
We refer to a distribution over allocations as a {\em faithful implementation} of $A^*$ if it is an implementation that satisfies small spread. 

\item {\em Ex-ante faithfulness.} In the randomized rounding, every agent $i$ gets in expectation value at least equal to her fractional value. $E[v_i(A_i)] \ge v_i(A^*_i)$. Observe that by definition, an implementation of the fractional allocation is ex-ante faithful.
 
\end{enumerate}

Faithful rounding of fractional solutions has a long history, where in different times researchers added additional ingredients (from those mentioned above) that they wished to satisfy. 
We briefly mention a few past relevant works.

Independent randomized rounding  has numerous applications for approximation algorithms. 
The rounding allocates each item to at most one agent, independently of the allocation of other items. 
That is, each item $j$ is independently (from other items) allocated to at most a single agent, with each agent $i$ getting item $j$ with probability equal to $A^*_{ij}$. 
This procedure provides a randomized polynomial time implementation for the fractional allocation (and hence is ex-ante faithful), but it does not provide ex-post faithfulness guarantees. 

Deterministic (polynomial time) rounding that is faithful from above was developed in~\citep{LST90} in the context of scheduling problems. For allocation problems, faithfulness from below is a more natural requirement, and this version was presented in~\citep{BD05}. A randomized polynomial time faithful implementation (showing that the small spread property holds and making explicit use it) was presented in~\citep{Srinivasan08}. A randomized polynomial time faithful implementation for a more general setting (referred to as a bi-hierarchy) was presented in~\citep{BCKM13}. Later work was concerned with deterministic (rather than randomized) polynomial time faithful implementations, with one approach described in~\citep{FSV20}, and a somewhat simpler approach presented in~\citep{Aziz20}. Summarizing the above discussion, and marginally improving over it (in terms of the upper bound on the support of the distribution), we have the following lemma.

\lemmaFaithful*


\begin{proof}
The proof of the lemma has two parts, neither one of them is new. The first (and main) part proves the lemma but without the upper bound of $f+1$, and the second part observes that standard techniques reduce the support to size $f+1$.

For the first part, we sketch for completeness the proof approach of~\citet{Aziz20}. Recall the Birkhoff -- von Neumann theorem that says that every doubly stochastic matrix can be decomposed into a weighted sum of permutation matrices. Equivalently, every perfect fractional matching in a bipartite graph can be decomposed into a weighted sum of perfect (integral) matchings. Moreover, this can be done in polynomial time, via repeatedly finding and peeling off a bipartite perfect matchings.

We reduce the setting of Lemma~\ref{lem:faithful} to that of the Birkhoff -- von Neumann theorem, {showing how we can take $A^*$ and generate a distribution over matchings of ``clones" of each agent, that can be use to generate a distribution over allocations that is a faithful implementation of $A^*$}. 
For every agent $i$ we do the following. Let $f_i= \sum_j A^*_{ij}$ denote the total sum of fractions of items (not their values) received by $i$ under $A^*$. We replace $i$ by $\lceil f_i \rceil$ {\em clones} $c_i^1, \ldots, c_i^{\lceil f_i \rceil}$ as follows. Sort all items in order of decreasing $v_i$ value. This gives a priority order for the following sequential ``eating" process. The clones of $i$ ``eat" the fractional allocation of $i$, where each clone in its turn consumes one unit of the fractional allocation ({starting consuming only after the prior clone completed consuming}), where the unit is chosen according to the priority order. The last clone might have less than a single unit to consume.

Having done the above for all agents, we now have a fractional matching between clones and items. This is not a perfect fractional matching (the last clone of an agent may consume less than one item), but the Birkhoff -- von Neumann theorem still applies (e.g., one can add dummy clones and items as needed so as to complete the instance to a perfect fractional matching on a larger bipartite graph).
Hence we can decompose the fractional matching into integral matchings. In every integral matching, every agent gets the items received by her  clones.

Ex-post faithfulness follows from the fact that for every agent $i$, in every integral allocation, each of $i$'s clones  (except for perhaps the last one) receives one item. Let $S_{i,\max}$ ($S_{i,\min}$, respectively) be the set of items obtained by taking for each of $i$'s clones the highest priority (lowest priority, respectively) item that the clone may possibly receive. 
Then every allocation that agent $i$ may receive has value in the range $[v_i(S_{i,\min}), v_i(S_{i,\max})]$. Observe that $v_i(S_{i,\min}) \ge v_i(S_{i,\max}) - \max_{j \in \items^f_i} v_i(j)$.
This last statement can be verified by removing the most valuable item (that of clone~1) from $S_{i,\max}$, and then using the fact that for every $j \le 1$, the item of clone $j$ in $S_{i,\min}$ is at least as valuable as the item of clone $j+1$ in $S_{i,\max}$. This established the small spread property, which implies ex-post faithfulness. 


The first part of the proof provided a deterministic polynomial time implementation of $A^*$ as a distribution $D$ over polynomially many allocations $A^1, A^2, \ldots A^{\ell}$, where every allocation in the support is ex-post faithful. In the second part we reduce the size of the support to $f+1$.  For this we set up a linear program. The variable $x_k$ specifies the extent to which we include allocation $A^k$ in the new implementation of $A^*$. The set $F$ contains those pairs $(i,j)$ for which in $A^*$, agent $i$ is allocated a strictly fractional part of item $j$, and $A^*_{i,j}$ denotes this fraction.  Observe that $|F| = f$. For every allocation $A^k$ in the support of $D$, we use $A^k_{i,j}$ as an indicator of whether item $j$ is allocated to agent $i$ in $A^k$. The $A^*_{i,j}$ and $A^k_{i,j}$ values serve as coefficients in our LP. The constraints of the LP are (the objective function can be set to~0):

\begin{enumerate}

\item $\sum_{1 \le k \le \ell} x_k = 1$.

\item $\sum_{i\in \agents, \; 1 \le k \le \ell} A^k_{i,j} x_k = A^*_{i,j}$ for every $(i,j) \in F$.

\item $x_k \ge 0$ for every $1 \le k \le \ell$. 

\end{enumerate}

The above LP is feasible, as the probabilities that $D$ assigns to each $A^k$ serve as a feasible solution. In polynomial time, one can find a basic feasible solution to the LP.
The number of non-zero variables in this solution is no larger than the number of constraints (excluding the non-negativity constraints), which is $f + 1$, as desired.
\end{proof}

		\section{Approximate fNSW}\label{app:fNSW}

{To demonstrate that fractional NSW approximation ex-ante is not in conflict with ex-post guarantees, we next show that in the simplest non-trivial cases (when there are at most $n$ items to allocate to $n$ additive agents\footnote{When there are less items than agents, the TPS is zero. In this case any randomized allocation that maximizes the fNSW is also proportional, and when supported on Pareto optimal allocations (using a faithful implementation) we get a stronger claim than the one proven in Theorem \ref{thm:matching} (the fNSW approximation is perfect).} ) 
it is indeed feasible to obtain fractional NSW approximation ex-ante,  while giving every agent at least her proportional share ex-ante and at least her TPS ex-post (with every ex-post allocation being Pareto optimal). 
It is easy to observe that these ex-post properties simply imply that the allocation is always a matching. Yet note that the agents have additive valuations, and the ex-ante benchmark of fNSW maximization is stronger than fNSW maximization  for unit-demand valuations (as valuations are additive).}

\begin{theorem}
\label{thm:matching}
For every instance with $n$ agents with additive valuations over 
at most $n$ items, there is a randomized allocation with the following properties.
\begin{itemize}

\item Ex-post: every agent gets at least her truncated proportional share (TPS), and the allocation is Pareto optimal (PO).

\item Ex-ante: every agent gets at least her proportional share (PS), and the randomized allocation approximates the fNSW with a ratio no worse than $e^{\frac{e+1}{e}} \simeq 3.927$.

\end{itemize}

Moreover, the fractional allocation associated with the randomized allocation can be computed in polynomial time.
\end{theorem}

\begin{proof}
Let $v_i$ denote the valuation function of agent $i$, naturally extended to fractional allocations. We assume without loss of generality that the number of items is exactly $n$ (which can be enforced by adding items of~0 value, if needed). Let $A = (A_1, \ldots, A_n)$ denote a fractional allocation. Let $|A_i|$ denote the sum of fractions of items allocated to agent $i$. Observe that a fractional allocation $A$ maximizes fNSW $(\prod_i v_i(A_i))^{\frac{1}{n}}$ if and only if it maximizes $\sum_i \log v_i(A_i)$. As the logarithm function is concave, the following optimization problem can be solved in polynomial time (up to arbitrary precision).

{\bf Maximize}  $\sum_i \log v_i(A_i)$ {\bf subject to}:

\begin{enumerate}

\item $A = (A_1, \ldots , A_n)$ is a fractional allocation.

\item $v_i(A_i) \ge \frac{1}{n}v_i(M)$. 

\item $|A_i| = 1$ for every $i$. 
\end{enumerate}

The above optimization problem is feasible (allocating to every agent a $\frac{1}{n}$ fraction of every item is a feasible solution).
By constraint~3, after rounding, every agent gets exactly one item. As there are $n$ items, the TPS of an agent is the value of the least valuable item for her, and hence every agent gets her TPS ex-post. The ex-post allocation need not be PO (e.g., an agent may receive an item of 0-value to him, that some other agent desires), but if needed, it can be replaced by a PO allocation that Pareto-dominates it (though we do not claim that this part can be done in polynomial time). By constraint~2, every agent gets at least her proportional share ex-ante. It remains to prove the constant approximation to the fNSW.

Let $A^*$ be a fractional allocation that maximizes fNSW. We show how it can be transformed into an allocation that satisfies the constraints of the optimization problem, while losing only a constant fraction in the value of the fNSW.

We first transform $A^*$ into an allocation $B$ that satisfies constraint~3. Let $P$ denote the set of those agents that under $A^*$ receive fractions adding up to strictly more than one item. Scale down the fractional allocation of every agent $i \in P$ by $|A^*_i|$. For the remaining agents, use the freed-up fractions of items to complete their allocation in an arbitrary way so that they each receive fractionally exactly one item. This gives the allocation $B$. 

We now compare $fNSW(A^*)$ with $fNSW(B)$. Observe that $fNSW(A^*) \le fNSW(B) \left(\prod_{i\in P} |A^*_i|\right)^{\frac{1}{n}}$. But as $\sum_{i\in P} |A^*_i| \le n$, a convexity argument shows that  $\left(\prod_{i\in P} |A^*_i| \right)^{\frac{1}{n}} \le \min_x (x^{\frac{n}{x}})^{\frac{1}{n}}$. The minimizer is $x = e$, and hence $fNSW(A^*) \le fNSW(B) \cdot e^{\frac{1}{e}}$. 

We now transform $B$ into an allocation $A$ that satisfies constraint~2. This is done in rounds. Starting at round $r=1$, we do the following. 

\begin{enumerate}

\item If the fractional allocation gives every agent at least her PS, then end and return this allocation as the allocation $A$.

\item Consider an arbitrary agent that does not receive her proportional share. Denote this agent by $Z_r$. 

\item For every agent other than $Z_1, \ldots Z_{r-1}$, scale its allocation by $\frac{n-r}{n-r+1}$. Observe that from every item, a $\frac{1}{n}$ fraction is now not allocated. 

\item Replace the allocation of $Z_r$ by an allocation that gives it a $\frac{1}{n}$ fraction of every item. Now $Z_r$ receives her proportional share.

\item {For every agent not in $Z_1,\ldots,Z_r$, use fractions of items freed-up by $Z_r$ to complete her sum of fractions of items to~1 in an arbitrary way.}
\end{enumerate} 

Note that the largest possible value of $r$ is $n-1$ (because if every agent in $Z_1, \ldots, Z_{n-1}$ gets a $\frac{1}{n}$ fraction of every item, so does the remaining agent).

We now compare $fNSW(A)$ with $fNSW(B)$. In a round $r$, every agent not in $\{Z_1, \ldots, Z_r\}$ maintains at least a  $\frac{n-r}{n-r+1}$ of her value, whereas agents in $\{Z_1, \ldots, Z_r\}$ do not lose value. Hence in the rounds leading to $r$, agent $Z_r$ maintains at least a $\frac{n-1}{n} \cdot \ldots \cdot \frac{n-r}{ n-r+1} = \frac{n-r}{n}$ of its value. Hence altogether, $$fNSW(B) \le fNSW(A) \left(\prod_{r=1}^{n-1} \frac{n}{n-r}\right)^{\frac{1}{n}}.$$ Observe that $$\prod_{r=1}^{n-1} \frac{n}{n-r} = \frac{n^{n-1}}{(n-1)!} \le \frac{n^{n-1} e^{n-1}}{(n-1)^{n-1} \sqrt{2\pi(n-1)}} = (1 + \frac{1}{n+1})^{n-1}\frac{e^{n-1}}{\sqrt{2\pi(n-1)}} \le e^n,$$
showing that $fNSW(B) \le fNSW(A) \cdot e.$
\end{proof}
	
\section{Combining EF1}\label{app:EF1}
An interesting challenge is to prove a result that adds EF1 to the guarantees we provide (proportionality and half the TPS). 
For two agents, running the standard cut-and-choose protocol with a random cutter gives a randomized allocation that is proportional ex-ante (which for two agents implies also EF ex-ante), and both EFX and MMS ex-post (and $2/3$-TPS). 
Unfortunately, we do not know of a randomized allocation that obtains such a result beyond two agents. Yet, below we prove a weaker result and present a randomized allocation that is EF1,  gives every agent at least a $\frac{n}{2n-1}$ fraction of her TPS, and $\frac{n}{2n-1}$ fraction of her proportional share ex-ante. Moreover, the allocation is poly-time computable. 
We next present the randomized allocation algorithm.


\begin{enumerate}

\item First allocate a single item to every agent as follows. Pick a random order over agents. According to that order, assign every agent an item of highest value to her among those items that are not yet assigned.  

\item \label{item2} Select an arbitrary agent that no one envies, and assign to that agent the item of highest value for her, among those items that have not yet been allocated. 

\item As long as there are envy cycles, eliminate them. Go back to step (\ref{item2}).

\end{enumerate}

\begin{theorem}\label{the:EF1-half-PS}
The randomized allocation produced by the above algorithm is EF1 and $\frac{n}{2n-1}$-TPS.
Additionally, it gives every agent at least a $\frac{n}{2n-1}$ fraction of her proportional share ex-ante.  Moreover, the allocation is poly-time computable. 
\end{theorem}
The claim that the allocation is poly-time computable is immediate. Thus, the theorem follows from the next two lemmas. 

\begin{lemma}
The allocation produced by the algorithm is EF1 and $\frac{n}{2n-1}$-TPS (ex-post, for every realization). 
\end{lemma}

\begin{proof}
The allocation is EF1 because for every agent $i$, all bundles start at equal value (of 0), and thereafter, there is no round in which a bundle that $i$ envies receives an item. Hence removing the last item in a bundle $B$ not allocated to $i$, agent $i$ does not envy that bundle. (The allocation is not necessarily EFX, as the last item in $B$ need not be the one of smallest value in the eyes of $i$.)

To see the TPS approximation, consider the first item $e$ received by $i$. If $v_i(e) \ge \frac{n}{2n-1} TPS_i$, we are done. If not, then observe that every bundle with higher value than $i$'s bundle either has one item, or its value exceeds that of $v_i(B_i)$ by at most $v_i(e)$. In the former case, eliminate the items and the agent, without hurting $TPS_i$. {The value received by $i$ is at least a $\frac{n'}{2n'-1}$ fraction of her proportional share of the set of those items that are in the $n'$ remaining bundles.}
\end{proof}

\begin{lemma}
For every number of agents, the uniformly random greedy algorithm gives every agent at least a $\frac{n}{2n-1}$ fraction of her proportional share ex-ante. 
\end{lemma}

\begin{proof}
If no item by itself is valued above the proportional share, then this follows from the ex-post $\frac{n}{2n-1}$-TPS guarantee. 

Hence, let $1 \le k < n$ denote the number of items of value above TPS, let $X$ denote their total value, and let $Y$ denote the total value of remaining items. Then with probability at least $\frac{k}{n}$ the agent gets one of the top items, and conditioned on that, the expected value received is at least $\frac{X}{k}$. With the remaining probability the agent gets at least $\frac{n}{2n-1}$-TPS, where the TPS is exactly $\frac{Y}{n-k}$. Hence in expectation the agent gets a least:

$$\frac{k}{n}\cdot \frac{X}{k} + \frac{n-k}{n}\cdot \frac{n}{2n-1}\cdot \frac{Y}{n-k} \ge \frac{n}{2n-1}\cdot \frac{X+Y}{n}$$ 
which is a $\frac{n}{2n-1}$ fraction of her proportional share.
\end{proof}

There are instances in which the allocation algorithm in the proof of Theorem~\ref{the:EF1-half-PS} does not provide every her proportional share ex-ante, not even if in Step~\ref{item2} of the algorithm the agent to receive an item is chosen uniformly at random among those agents that no one envies. For example, if there are five items and valuations $(12, 10, 9, 8, 5)$ and $(10, 12, 9, 6, 5)$ then the proportional share of the first agent is $22$, this randomized version of the greedy algorithm will give her the bundle $(12, 9)$ with probability $\frac{1}{2}$, and each of the bundles $(12, 8)$ and $(12, 8, 5)$ with probability $\frac{1}{4}$, for a total expected value of $21.75$. Adding a large constant $M$ to the value of every item, the proportional share of the agent is roughly $\frac{5}{2}M$, whereas in expectation she gets $\frac{9}{4}M$, showing that one does not get better than a $\frac{9}{10}$ approximation to the proportional share ex-ante.

Unfortunately, the allocation need not give more than $\frac{n}{2n-1}$-MMS ex-post.
An example showing that the allocation need not give more than $\frac{n}{2n-1}$-MMS is as follows. For agent~1, the first $n$ items have value $\frac{n}{2n-1}$, and all remaining items have value small~$\epsilon>0$, where the proportional share and MMS are~1. For other agents, the first item has value~1, and the remaining item have value $\epsilon' > 0$,  where the proportional share and MMS are~1. If agent~1 gets the first item, the other agents might get all remaining items in reverse order (items~2 to~$n$ are allocated last, each to a different agent).

Our results imply that every two of the properties ex-post EF1, ex-post approximate TPS and ex-ante proportional can be achieved simultaneously. The combination of ex-post EF1 and $\frac{n}{2n-1}$-TPS follows from Theorem~\ref{the:EF1-half-PS}, the combination of ex-post $\frac{1}{2}$-TPS and ex-ante proportional follows from Theorem~\ref{thm:BoB}, and the combination of ex-post EF1 and ex-ante proportional follows from 
\cite{Aziz20}. The question of whether all three can be achieved simultaneously remains open.

\end{document}